%% file: sharingJournalclean.tex
\documentclass[draftclsnofoot,onecolumn,12pt]{IEEEtran}
\newtheorem{theorem}{Theorem}
\newtheorem{lemma}{Lemma}
\newtheorem{remark}{Remark}
\newtheorem{definition}{Definition}
\newtheorem{corollary}{Corollary}

\newtheorem{claim}{Claim}
\newtheorem{example}{Example}
\usepackage{graphicx}
\usepackage{cite}
\usepackage{verbatim}
\usepackage{dsfont}
\usepackage{amsmath, amssymb}
\graphicspath{{figs/}}
\usepackage[caption=false,font=footnotesize]{subfig}

\usepackage{fixltx2e}
\usepackage{pseudocode}
\usepackage{algorithmic,algorithm}

\usepackage[final,ulem=normalem]{changes}
\definechangesauthor{RW}{blue}
\definechangesauthor{TC}{red}

\begin{document}

\title{Coded Cooperative Data Exchange\\ in Multihop Networks}

\author{Thomas~A.~Courtade,~\IEEEmembership{Student~Member,~IEEE,}
        and~Richard~D.~Wesel,~\IEEEmembership{Senior~Member,~IEEE}
\thanks{The authors are with the Electrical Engineering Department, University of California, Los Angeles, CA, 90095 USA (email: tacourta@ee.ucla.edu; wesel@ee.ucla.edu).}
\thanks{This work was presented in part at the 2010 Military Communications Conference (MILCOM) \cite{bib:CourtadeWeselMILCOM2010} and the 2010-2011 Allerton Conference on Communication, Control, and Computing \cite{bib:CourtadeWeselAllerton2010, bib:CourtadeWeselAllerton2011}.}
\thanks{This research was supported by Rockwell Collins through contract \#4502769987.}
}
 
\maketitle

\begin{abstract}
Consider a connected network of $n$ nodes that all wish to recover $k$ desired packets.  Each node begins with a subset of the desired packets and exchanges coded packets with its neighbors.  This paper provides necessary and sufficient conditions which characterize the set of all transmission schemes that permit 
every node to ultimately learn (recover) all $k$ packets.  When the network satisfies certain regularity conditions and packets are randomly distributed, this paper provides tight concentration results on the number of transmissions required to achieve universal recovery.  For the case of a fully connected network, a polynomial-time algorithm for computing an optimal transmission scheme is derived.  An application to secrecy generation is discussed.

\end{abstract}


\section{Introduction}\label{sec:Intro}
Consider a connected network of $n$ nodes that all wish to recover $k$ desired packets.  Each node begins with a subset of the desired packets and broadcasts messages to its neighbors over discrete, memoryless, and interference-free channels.  Furthermore, every node knows which packets are already known by each node and knows the topology of the network.  How many transmissions are required to disseminate the $k$ packets to every node in the network?  How should this be accomplished?  These are the essential questions addressed.  We refer to this as the \emph{Coded Cooperative Data Exchange} problem, or just the Cooperative Data Exchange problem.

This work is motivated in part by emerging issues in distributed data storage.  Consider the problem of backing up data on servers in a large data center.  One commonly employed method to protect data from corruption is  replication.  Using this method, large quantities of data are replicated in several locations so as to protect from various sources of corruption (e.g., equipment failure, power outages, natural disasters, etc.).  As the quantity of information in large data centers continues to increase, the number of file transfers required to complete a periodic replication task is becoming an increasingly important consideration due to time, equipment, cost, and energy constraints.  The results contained in this paper address these issues.

This model also has natural applications in the context of tactical networks, and we give one of them here.  Consider a scenario in which an aircraft flies over a group of nodes on the ground and tries to deliver a video stream.  Each ground node might only receive a subset of the transmitted packets due to interference, obstructions, and other signal integrity issues.  In order to recover the transmission, the nodes are free to communicate with their neighbors, but would like to minimize the number of transmissions in order to conserve battery power (or avoid detection, etc.).  How should the nodes share information, and what is the minimum number of transmissions required so that the entire network can recover the video stream? 

Beyond the examples mentioned above, the results presented herein can also be applied to practical secrecy generation amongst a collection of nodes.  We consider this application in detail in Section \ref{sec:URSecrecy}.

\subsection{Related Work}

Distributed data exchange problems have received a great deal of attention over the past several years.  The powerful techniques afforded by network coding \cite{bib:Ahlswede00networkinformation, bib:Li03linearnetwork} have paved the way for cooperative communications at the packet-level.

The coded cooperative data exchange problem (also called the universal recovery problem in \cite{bib:CourtadeWeselMILCOM2010,bib:CourtadeWeselAllerton2010,bib:CourtadeWeselAllerton2011}) was originally introduced by El Rouayheb {\it et al.} in \cite{bib:ElRouayhebChaudhryITW2007, bib:ElRouayheb2010ITW } for a fully connected network (i.e., a single-hop network).   For this special case, a randomized algorithm for finding an optimal transmission scheme was given in \cite{bib:SprintsonSadeghiISIT2010}, and the first deterministic algorithm was recently given in \cite{bib:SprintsonQshine2010}. In the concluding remarks of \cite{bib:SprintsonQshine2010}, the weighted universal recovery problem (in which the objective is to minimize the weighted sum of transmissions by nodes) was posed as an open problem. However, this was solved using a variant of the same algorithm in \cite{bib:OzgulSprintsonITA2011}, and independently by the present authors using a submodular algorithm in \cite{bib:CourtadeWeselAllerton2011}.  

The coded cooperative data exchange problem is related to the index coding problem originally introduced by Birk and Kol in \cite{bib:BirkKolIT2006}.  Specifically, generalizing the index coding problem to permit each node to be a transmitter (instead of having a single server) and further generalizing so that the network need not be a single hop network leads to a class of problems that includes our problem as a special case in which each node desires to receive all packets.

One  significant result in index coding is that nonlinear index coding outperforms the best linear index code in certain cases  \cite{bib:LubetzkyStavIT2009,
bib:AlonLubetzkyFOCS2008}. As discussed above, our problem is a special case of the generalized index coding problem, and it turns out that linear encoding does achieve the minimum number of transmissions required for universal recovery and  this solution is computable in polynomial time for some important cases.

This paper applies principles of cooperative data exchange to generate secrecy in the presence of an eavesdropper.  In this context, the secrecy generation problem was originally studied in \cite{bib:CsiszarNarayanIT2004}.  In \cite{bib:CsiszarNarayanIT2004}, Csiszar and Narayan gave single-letter characterizations of the secret-key and private-key capacities for a network of nodes connected by an error-free broadcast channel.  While  general and powerful, these results left two practical issues as open questions.  First, (as with many information-theoretic investigations) the results require the nodes to observe arbitrarily long sequences of i.i.d.~source symbols, which is generally not practical.  Second, no efficient algorithm is provided in \cite{bib:CsiszarNarayanIT2004} which achieves the respective secrecy capacities.  More recent work in \cite{bib:YeNarayanISIT2005,bib:YeReznikCISS2010}  addressed the latter point.

\subsection{Our Contributions}

In this paper, we provide necessary and sufficient conditions for achieving universal recovery\footnote{In this paper, we use the term \emph{universal recovery} to refer to the ultimate condition where every node has successfully recovered all packets.} in arbitrarily connected multihop networks.  
We specialize these necessary and sufficient conditions to obtain precise results in the case where the underlying network topology satisfies some modest regularity conditions.  

For the case of a fully connected network, we provide an algorithm based on submodular optimization which solves the cooperative data exchange problem.  This algorithm is unique from the others previously appearing in the literature (cf. \cite{bib:SprintsonSadeghiISIT2010,bib:SprintsonQshine2010,bib:OzgulSprintsonITA2011}) in that it exploits submodularity.  As a corollary, we provide exact concentration results when packets are randomly distributed in a network.  

In this same vein, we  also   obtain tight concentration results and approximate solutions when the underlying network is $d$-regular and packets are distributed randomly.

Furthermore, if packets are divisible (allowing transmissions to consist of partial packets), we prove that the traditional cut-set bounds can be achieved for any network topology.  In the case of $d$-regular and fully connected networks, we show that splitting packets does not typically provide any significant benefits.

Finally, for the application to secrecy generation, we leverage the results of \cite{bib:CsiszarNarayanIT2004} in the context of the cooperative data exchange problem for a fully connected network.  In doing so, we provide an efficient algorithm that achieves the secrecy capacity without requiring any quantities to grow asymptotically large.

\subsection{Organization}

This paper is organized as follows.  Section \ref{sec:ModelBcast} formally introduces the problem and provides basic definitions and notation.  Section \ref{sec:URMainResults} presents our main results. Section \ref{sec:URSecrecy} discusses the application of our results to secrecy generation by a collection of nodes in the presence of an eavesdropper.  Section \ref{sec:URProofs} contains the relevant proofs.   Section \ref{sec:URConclusion} delivers the conclusions and discusses directions for future work.

\section{System Model and Definitions}\label{sec:ModelBcast}
Before we formally introduce the problem, we establish some notation.  Let $\mathbb{N}={0,1,2,\dots}$ denote the set of natural numbers.
For two sets $A$ and $B$, the relation $A \subset B$ implies that $A$ is a proper subset of $B$ (i.e., $A\subseteq B$ and $A\neq B$).  For a set $A$, the corresponding power set is denoted $2^A:=\{B : B \subseteq A\}$. We  use the notation $[m]$ to denote the set $\{1,\dots,m\}$.

This paper considers a network  of $n$ nodes.  The network must be connected, but it need not be fully connected (i.e., it need not be a complete graph).  A graph $\mathcal{G}=(V,E)$ describes the specific connections in the network, where $V$ is the set of vertices $\{v_i: i \in \{1, \ldots, n\} \}$ (each corresponding to a node)  and $E$ is the set of edges connecting nodes. We assume that the edges in $E$ are undirected, but our results can be extended to directed graphs.  

Each node wishes to recover the same $k$ desired packets, and each node begins with a (possibly empty) subset of the desired packets.  Formally, let ${P}_i\subseteq \{p_1,\dots,p_k\}$ be the (indexed) set of packets originally available at node $i$, and $\{{P}_i\}_{i=1}^n$ satisfies $\bigcup_{i=1}^n {P}_i=\{p_1,\dots,p_k\}$.  Each $p_j\in\mathbb{F}$, where $\mathbb{F}$ is some finite field (e.g. $\mathbb{F}=\mbox{GF}(2^m)$).  For our purposes, it suffices to assume $|\mathbb{F}|\geq 2n$.  The set of packets initially missing at node $i$ is denoted $P_i^c:=\{p_1,\dots,p_k\}\backslash P_i$.

Throughout this paper, we assume that each packet $p_i\in \{p_1,\dots,p_k\}$ is equally likely to be any element of $\mathbb{F}$.  Moreover, we assume that packets are independent of one another.  Thus, no correlation between different packets or prior knowledge about unknown packets can be exploited.

To simplify notation, we will refer to a given problem instance (i.e., a graph and corresponding sets of packets available at each node) as a network $\mathcal{T}=\{\mathcal{G},P_1,\dots,P_n\}$.  When no ambiguity is present, we will refer to a network by $\mathcal{T}$ and omit the implicit dependence on the parameters $\{\mathcal{G},P_1,\dots,P_n\}$.

Let the set $\Gamma(i)$ be the neighborhood of node $i$. There exists an edge $e\in E$ connecting two vertices $v_i,v_j\in V$ iff $i\in \Gamma(j)$.   For convenience, we put $i\in \Gamma(i)$.  Node $i$ sends (possibly coded) packets to its neighbors $\Gamma(i)$ over discrete, memoryless, and interference-free channels.  In other words, if node $i$ transmits a message, then every node in $\Gamma(i)$ receives that message. If $S$ is a set of nodes, then we define $\Gamma({S})=\cup_{i\in {S}} \Gamma(i)$.  In a similar manner, we define $\partial(S)=\Gamma(S)\backslash S$ to be the boundary of the vertices in $S$.  An example of sets $S$, $\Gamma(S)$, and $\partial(S)$ is given in Figure \ref{fig:G_sets}.

\begin{figure}
\centering
\def\svgwidth{3in}
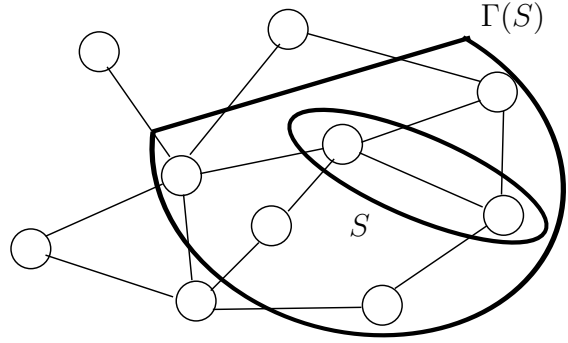
\caption{For the given graph, a set of vertices $S$ and its neighborhood $\Gamma(S)$ are depicted.  The set $\partial(S)$ (i.e., the boundary of $S$) consists of the four vertices in $\Gamma(S)$ which are not in $S$.}
\label{fig:G_sets}
\end{figure}

This paper seeks to determine the minimum number of transmissions required to achieve universal recovery (when every node has learned all $k$ packets).  We primarily consider the case where packets are deemed indivisible.  In this case, a single transmission by user $i$ consists of sending a packet (some $z\in\mathbb{F}$) to all nodes $j \in \Gamma(i)$.  This motivates the following definition.

\begin{definition}\label{def:URminNumberM}
Given a network $\mathcal{T}$, the minimum number of transmissions required to achieve universal recovery is denoted $M^*(\mathcal{T})$.
\end{definition}

To clarify this concept, we briefly consider two examples:

\begin{example}[Line Network] \label{ex:B_N_line}
Suppose $\mathcal{T}$ is a network of nodes connected along a line as follows: $V=\{v_1,v_2,v_3\}$, $E=\{(v_1,v_2),(v_2,v_3)\}$, ${P}_1=\{p_1\}$, ${P}_2=\emptyset$, and ${P}_3=\{p_2\}$.  Note that each node must transmit at least once in order for all nodes to recover  $\{p_1,p_2\}$, hence $M^*(\mathcal{T})\geq 3$.  Suppose node 1 transmits $p_1$ and node 3 transmits $p_2$. Then (upon receipt of $p_1$ and $p_2$ from nodes 1 and 3, respectively) node 2 transmits $p_1\oplus p_2$,  where $\oplus$ indicates addition in the finite field $\mathbb{F}$.  This strategy requires 3 transmissions and allows each user to recover $\{p_1,p_2\}$.  Hence $M^*(\mathcal{T})= 3$.
\end{example}

Example \ref{ex:B_N_line} demonstrates a transmission schedule that uses two \emph{rounds} of communication.  The transmissions by node $i$ in a particular round of communication can depend only on the information available to node $i$ prior to that round (i.e. ${P}_i$ and previously received transmissions from neighboring nodes).  In other words, the transmissions are causal.  The transmission scheme employed in Example \ref{ex:B_N_line} is illustrated in Figure \ref{fig:lineExFig}.

\begin{figure}
\centering
\def\svgwidth{5in}
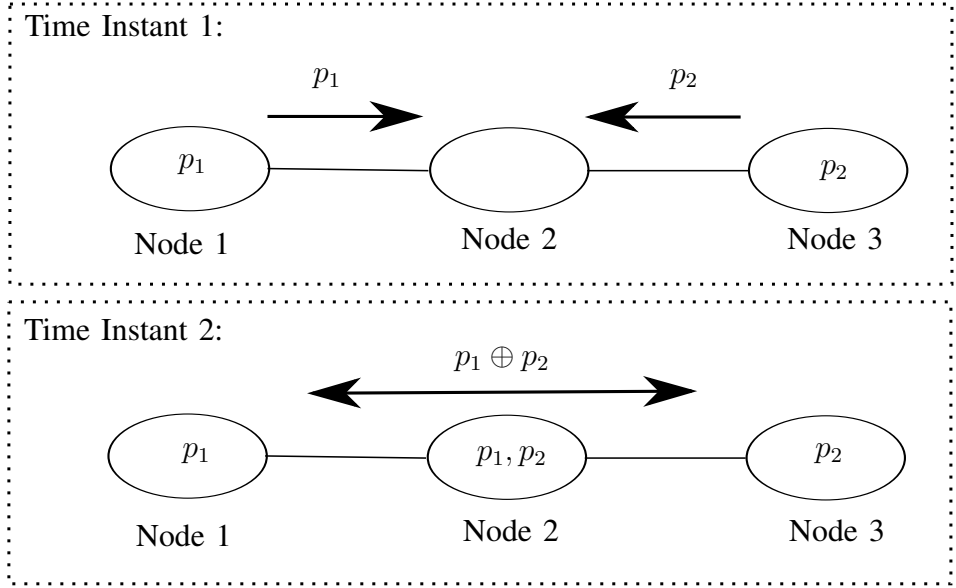
\caption{An illustration of the transmission scheme employed in Example \ref{ex:B_N_line}.  During the first time instant, Nodes 1 and 3 broadcast packets $p_1$ and $p_2$, respectively. During the second time instant, Node 2 broadcasts the XOR of packets $p_1$ and $p_2$.  This scheme requires three transmissions and achieves universal recovery.}
\label{fig:lineExFig}
\end{figure}

\begin{example}[Fully Connected Network]\label{ex:B_N_clique}
Suppose $\mathcal{T}$ is a 3-node fully connected network in which  $\mathcal{G}$ is a complete graph on 3 vertices, and  ${P}_i=\{p_1,p_2,p_3\} \backslash p_i$.  Clearly one transmission is not sufficient, thus $M^*(\mathcal{T})\geq 2$.  It can be seen that two transmissions suffice: let node 1 transmit $p_2$ which lets node 2 have ${P}_2 \cup p_2 = \{p_1,p_2,p_3\}$.  Now, node 2 transmits $p_1\oplus p_3$, allowing nodes 1 and 3 to each recover all three packets.  Thus $M^*(\mathcal{T})= 2$.  Since each transmission was only a function of the packets originally available at the corresponding node, this transmission strategy can be accomplished in a single round of communication.
\end{example}

In the above examples, we notice that the transmission schemes are partially characterized by a schedule of which nodes transmit during which round of communication.  We formalize this notion with the following definition:

\begin{definition}[Transmission Schedule] A set of integers $\{b_i^{j} : i\in [n],j\in [r], b_i^j\in \mathbb{N} \}$ is called a transmission schedule for $r$ rounds of communication if node $i$ makes exactly $b_i^j$ transmissions during communication round $j$.  \end{definition}

When the parameters $n$ and $r$ are clear from context, a transmission schedule will be denoted by the shorthand notation $\{b_i^j\}$.  Although finding a transmission schedule that achieves universal recovery is relatively easy (e.g., each node transmits all packets in their possession at each time instant), finding one that achieves universal recovery with $M^*(\mathcal{T})$ transmissions can be extremely difficult.  This is demonstrated by the following example:

\begin{example}[Optimal Cooperative Data Exchange is NP-Hard.]\label{ex:hardness}
Suppose $\mathcal{T}$ is a network with $k=1$ corresponding to a bipartite graph with left and right vertex sets $V_L$ and $V_R$ respectively.  Let $P_i=p_1$ for each $i\in V_L$, and let $P_i=\emptyset$ for each $i\in V_R$.  In this case, $M^*(\mathcal{T})$ is given by the minimum number of sets in $\{\Gamma(i)\}_{i\in V_L}$ which cover all vertices in $V_R$.  Thus, finding $M^*(\mathcal{T})$ is at least as hard as the Minimum Set Cover problem, which is NP-complete \cite{bib:Karp1972}.
\end{example}

Several of our results are stated in the context of \emph{randomly distributed packets}.  Assume $0<q<1$ is given.  Our model is essentially that each packet is available independently at each node with probability $q$.  However, we must condition on the event that each packet is available to at least one node.  Thus, when packets are randomly distributed, the underlying probability measure is given by
\begin{align}
\Pr\left[ p_i \in \bigcup_{j\in S} P_j \right] = \frac{1-(1-q)^{|S|}}{1-(1-q)^n}\label{eqn:probModel}
\end{align}
for all $i\in [k]$ and all nonempty $S\subseteq V =[n]$.

Finally, we introduce one more definition which links the network topology with the number of communication rounds, $r$.

\begin{definition}
For a graph $\mathcal{G}=(V,E)$ on $n$ vertices, define $\mathcal{S}^{(r)}(\mathcal{G})\subset (2^V)^{r+1}$ as follows: $(S_0,S_1,\dots,S_r) \in  \mathcal{S}^{(r)}(\mathcal{G})$ if and only if the sets $\{S_i\}_{i=0}^r$ satisfy the following two conditions:
\begin{align*}
\emptyset \subset S_i \subset V &\mbox{~~for each $0\leq i \leq r$, and} \\
S_{i-1} \subseteq S_{i} \subseteq \Gamma(S_{i-1}) &\mbox{~~for each $1\leq i \leq r$.}
\end{align*}
\end{definition}
In words, any element in $\mathcal{S}^{(r)}(\mathcal{G})$ is a nested sequence of subsets of vertices of $\mathcal{G}$.  Moreover, the constraint that each set in the sequence is contained in its predecessor's neighborhood implies that the sets cannot expand too quickly relative to the topology of $\mathcal{G}$. 

To make the definition of $\mathcal{S}^{(r)}(\mathcal{G})$ more concrete, we have illustrated a sequence $(S_0,S_1,S_2)\in\mathcal{S}^{(2)}(\mathcal{G})$ for a particular choice of graph $\mathcal{G}$ in Figure \ref{fig:Sets}.
\begin{figure}
\centering
\def\svgwidth{4in}
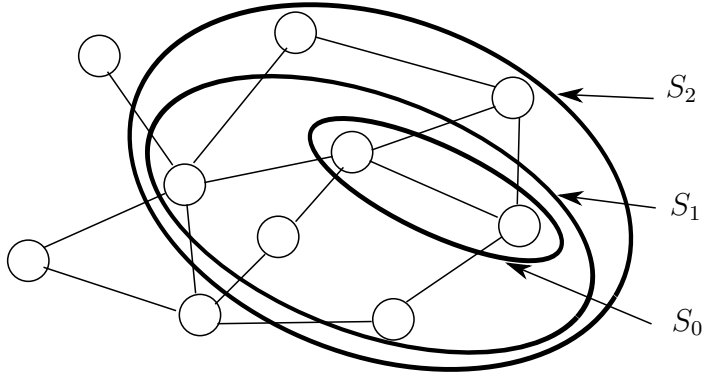
\caption{An example of a sequence $(S_0,S_1,S_2)\in\mathcal{S}^{(2)}(\mathcal{G})$ for a particular choice of graph $\mathcal{G}$.}
\label{fig:Sets}
\end{figure}

\section{Main Results}\label{sec:URMainResults}
In this section, we present our main results.  Proofs are delayed until Section \ref{sec:URProofs}.

\subsection{Necessary and Sufficient Conditions for Universal Recovery}
First, we provide necessary and sufficient conditions for achieving universal recovery in a network $\mathcal{T}$.  It turns out that these conditions are characterized by a particular set of transmission schedules $\mathcal{R}_r(\mathcal{T})$ which we define as follows:

\begin{definition}
For a network $\mathcal{T}=\{\mathcal{G},P_1,\dots,P_n\}$, define the region $\mathcal{R}_r(\mathcal{T}) \subseteq \mathbb{N}^{n \times r}$  to be the set of all transmission schedules $\{b_i^j\}$ satisfying:
\begin{align*}
\sum_{j=1}^r ~~& \sum_{i\in S_j^c \cap \Gamma(S_{j-1})} b^{(r+1-j)}_i \geq \left| \bigcap_{i\in S_r} P_i^c \right| \mbox{~~for each $(S_0,\dots,S_r) \in \mathcal{S}^{(r)}(\mathcal{G})$}.
\end{align*}
\end{definition}

\begin{theorem} \label{thm:generalFeasibleRegion}
For a network $\mathcal{T}$, a transmission schedule $\{b_i^j\}$ permits universal recovery in $r$ rounds of communication if and only if $\{b_i^j\}\in \mathcal{R}_r(\mathcal{T})$.
\end{theorem}

Theorem \ref{thm:generalFeasibleRegion} reveals that the set of transmission schedules permitting universal recovery is characterized precisely by the region $\mathcal{R}_r(\mathcal{T})$.  In fact, given a transmission schedule in $\mathcal{R}_r(\mathcal{T})$, a corresponding coding scheme that achieves universal recovery can be  computed in polynomial time using the algorithm in \cite{bib:JaggiIT2005} applied to the network coding graph discussed in the proof of Theorem \ref{thm:generalFeasibleRegion}.  
Alternatively, one could employ random linear network coding over a sufficiently large field size \cite{bib:randLinNC}.  If transmissions are made in a manner consistent with a schedule in $\mathcal{R}_r(\mathcal{T})$, universal recovery will be achieved with high probability.

Thus, the problem of achieving universal recovery with the minimum number of transmissions reduces to solving a combinatorial optimization problem over $\mathcal{R}_r(\mathcal{T})$.  As this problem was shown to be NP-hard in Example \ref{ex:hardness}, we do not attempt to solve it in its most general form.  Instead, we apply Theorem \ref{thm:generalFeasibleRegion} to obtain surprisingly simple characterizations for several cases of interest.  

Before proceeding, we provide a quick example showing how the traditional cut-set bounds can be recovered from Theorem \ref{thm:generalFeasibleRegion}.

\begin{example}[Cut-Set Bounds]
Considering the constraint defining $\mathcal{R}_r(\mathcal{T})$ in which the nested subsets that form $\mathcal{S}^{(r)}(\mathcal{G})$ are all identical.  That is, $(S,S,\dots,S)\in \mathcal{S}^{(r)}(\mathcal{G})$ for some nonempty ${S}\subset V$.  We see that any transmission schedule $\{b_i^j\}\in \mathcal{R}_r(\mathcal{T})$ must satisfy the familiar cut-set bounds:
\begin{align}
\sum_{j=1}^r \sum_{i\in \partial({S})} &  b^{j}_i \geq \left| \bigcap_{i\in {S}} {P}_i^c \right|. \label{eqn:cutsetBounds}
\end{align}
In words, the total number of packets that flow {into} the set of nodes ${S}$ must be greater than or equal to the number of packets that the nodes in {$S$} are collectively missing.
\end{example}

\subsection{Fully Connected Networks}
When $\mathcal{T}$ is a fully connected network, the graph $\mathcal{G}$ is a complete graph on $n$ vertices.  This is perhaps one of the most practically important cases to consider.  For example, in a wired computer network, clients can multicast their messages to all other terminals which are cooperatively exchanging data.  In wireless networks, broadcast is a natural transmission mode.  Indeed, there are  protocols tailored specifically to wireless networks which support reliable network-wide broadcast capabilities (cf. \cite{bib:HalfordChuggPIMRC2010, bib:HalfordChuggITA2010, bib:HalfordHwangMilcom2010, bib:BlairBrownChuggHalfordMILCOM2008}).  It is fortunate then, that the cooperative data exchange problem can be solved in polynomial time for fully connected networks:

\begin{theorem} \label{thm:URFullyConnectedPolyTime}
For a fully connected network $\mathcal{T}$, a transmission schedule requiring only $M^*(\mathcal{T})$ transmissions can be computed in polynomial time.  Necessary and sufficient conditions for universal recovery in this case are given by the cut-set constraints \eqref{eqn:cutsetBounds}.  
Moreover, a single round of communication is sufficient to achieve universal recovery with $M^*(\mathcal{T})$ transmissions.
\end{theorem}

For the fully connected network in Example \ref{ex:B_N_clique}, we remarked that only one round of transmission was required.  Theorem \ref{thm:URFullyConnectedPolyTime} states that this trend extends to any fully connected network.

An algorithm for solving the cooperative data exchange problem for fully connected networks is presented in Appendix \ref{sec:Algo}.  We remark that the algorithm is sufficiently general that it can also solve the cooperative data exchange problem where the objective is to minimize the weighted sum of nodes' transmissions.

Although Theorem \ref{thm:URFullyConnectedPolyTime} applies to arbitrary sets of packets $P_1,\dots,P_n$, it is insightful to consider the case where packets are randomly distributed in the network.  In this case, the minimum number of transmissions required for universal recovery converges in probability to a simple function of the (random) sets $P_1,\dots,P_n$.
\begin{theorem} \label{thm:URcliqueRandom}
If $\mathcal{T}$ is a fully connected network and packets are randomly distributed, then
\begin{align*}
M^*(\mathcal{T}) = \left\lceil \frac{1}{n-1} \sum_{i=1}^n |P_i^c| \right\rceil.
\end{align*}
with probability approaching $1$ as the number of packets $k\rightarrow \infty$.
\end{theorem}

\subsection{$d$-Regular Networks}
Given that  precise results can be obtained for fully connected networks, it is natural to ask whether these results can be extended to a larger class of networks which includes fully connected networks as a special case.  In this section, we partially answer this question in the affirmative.  To this end, we define  $d$-regular networks.

\begin{definition}[$d$-Regular Networks]
A network $\mathcal{T}$ is said to be $d$-regular if $\partial(i)=d$ for each $i\in V$ and $\partial(S)\geq d$ for each nonempty $S\subset V$ with $|S|\leq n-d$.  In other words, a network $\mathcal{T}$ is $d$-regular if the associated graph $\mathcal{G}$ is $d$-regular and $d$-vertex-connected.
\end{definition} 

Immediately, we see that the class of $d$-regular networks includes fully connected networks as a special case with $d=n-1$.  Further, the class of $d$-regular networks includes many frequently studied network topologies (e.g., cycles, grids on tori, etc.).

Unfortunately, the deterministic algorithm of Theorem \ref{thm:URFullyConnectedPolyTime} does not appear to extend to $d$-regular networks.  However, a slightly weaker concentration result similar to Theorem \ref{thm:URcliqueRandom} can be obtained when packets are randomly distributed.  Before stating this result, consider the following Linear Program (LP) with variable vector $x\in \mathbb{R}^n$ defined for a network $\mathcal{T}$:
\begin{align}
\mbox{minimize~~} &\sum_{i =1}^n x_i  \label{lp:dregObj} \\
\mbox{subject to:~~} & \sum_{i\in \partial(j)} x_i \geq \left| P_j^c \right| \mbox{~~for each $j\in V$}. \label{lp:dregConst}
\end{align} 
Let $M_{LP}(\mathcal{T})$ denote the optimal value of this LP.  Interpreting $x_i$ as $\sum_j b_i^j$, the constraints in the LP are a subset of the cut-set constraints of (\ref{eqn:cutsetBounds}) which are a subset of the necessary constraints for universal recovery given in Theorem \ref{thm:generalFeasibleRegion}.  Furthermore, the integer constraints on the $x_i$'s are relaxed. Thus $M_{LP}(\mathcal{T})$ certainly bounds $M^*(\mathcal{T})$ from below.  Surprisingly, if $\mathcal{T}$ is a $d$-regular network and the packets are randomly distributed, $M^*(\mathcal{T})$ is very close to this lower bound with high probability:
\begin{theorem}\label{thm:URdReg}
If $\mathcal{T}$ is a $d$-regular network and the packets are randomly distributed, then
\begin{align*}
M^*(\mathcal{T}) < M_{LP}(\mathcal{T})+n
\end{align*}
with probability approaching $1$ as the number of packets $k\rightarrow \infty$.
\end{theorem}

We make two important observations.  First, the length of the interval in which $M^*(\mathcal{T})$ is concentrated is independent of $k$.  Hence, even though the number of packets $k$ may be extremely large, $M^*(\mathcal{T})$ can be estimated accurately.  Second, as $k$ grows large, $M^*(\mathcal{T})$ is dominated by the \emph{local} topology of $\mathcal{T}$.  This is readily seen since the constraints defining $M_{LP}(\mathcal{T})$ correspond only to nodes' immediate neighborhoods.  The importance of the local neighborhood  was also seen in \cite{bib:RamamoorthyIT2005} where network coding capacity for certain random networks is shown to concentrate around the expected number of nearest neighbors of the source and the terminals.

\subsection{Large (Divisible) Packets}
We now return to general networks with arbitrarily distributed packets.  However, we now consider the case where packets are ``large" and can be divided into several smaller pieces (e.g., packets actually correspond to large files).  To formalize this, assume that each packet can be partitioned into $t$ chunks of equal size, and transmissions can consist of a single chunk (as opposed to an entire packet).  In this case, we say the packets are $t$-divisible.  To illustrate this point more clearly, we return to Example \ref{ex:B_N_clique}, this time considering $2$-divisible packets.

\begin{example}[$2$-Divisible Packets]\label{ex:B_t_N}
Let $\mathcal{T}$ be the network of Example \ref{ex:B_N_clique} and split each packet into two halves: $p_i \rightarrow (p_i^{(1)},p_i^{(2)})$.  Denote this new network $\mathcal{T}'$ with corresponding sets of packets:
\begin{align*}
{P}_i'=\{p_1^{(1)},p_1^{(2)}p_2^{(1)},p_2^{(2)},p_3^{(1)},p_3^{(2)}\} \backslash \{p_i^{(1)},p_i^{(2)}\}.
\end{align*}
Three chunk transmissions allow universal recovery as follows:  Node 1 transmits $p_2^{(2)} \oplus p_3^{(2)}$.  Node 2 transmits $p_1^{(1)} \oplus p_3^{(1)}$.  Node 3 transmits $p_1^{(2)} \oplus p_2^{(1)}$.  It is readily verified from (\ref{eqn:cutsetBounds}) that $3$ chunk-transmissions are required to permit universal recovery.  Thus, $M^*(\mathcal{T}')=3$.  Hence, if we were allowed to split the packets of Example \ref{ex:B_N_clique} into two halves, it would suffice to transmit 3 chunks.  Normalizing the number of transmissions by the number of chunks per packet, we say that universal recovery can be achieved with $1.5$ packet transmissions.
\end{example}

Motivated by this example, define $M^*_t(\mathcal{T})$ to be the minimum number of (normalized) packet-transmissions required to achieve universal recovery in the network $\mathcal{T}$ when packets are $t$-divisible.  For the network $\mathcal{T}$ in Example \ref{ex:B_N_clique}, we saw above that $M^*_2(\mathcal{T})=1.5$.

It turns out, if packets are $t$-divisible and $t$ is large, the cut-set bounds \eqref{eqn:cutsetBounds} are ``nearly sufficient" for achieving universal recovery.  To see this, let $M_{\mbox{cut-set}}(\mathcal{T})$ be the optimal value of the LP:
\begin{align}
\mbox{minimize~~} &\sum_{i =1}^n x_i \label{lp:dvisObj}\\
\mbox{subject to:~~} & \sum_{i\in \partial(S)} x_i \geq \left| \bigcap_{i\in S} P_i^c \right| \mbox{~~for each nonempty $S \subset V$}.\label{lp:divisConst}
\end{align} 
Clearly $M_{\mbox{cut-set}}(\mathcal{T})\leq M^*_t(\mathcal{T})$ for any network $\mathcal{T}$ with $t$-divisible packets because the LP producing $M_{\mbox{cut-set}}(\mathcal{T})$ relaxes the integer constraints and is constrained only by \eqref{eqn:cutsetBounds} rather than the full set of constraints given in Theorem \ref{thm:generalFeasibleRegion}.  However, there exist transmission schedules which can approach this lower bound.  Stated more precisely:
\begin{theorem}\label{thm:URtsplit}
For any network $\mathcal{T}$, the minimum number of (normalized) packet-transmissions required to achieve universal recovery with $t$-divisible packets satisfies
\begin{align*}
\lim_{t\rightarrow \infty} M^*_t(\mathcal{T}) =M_{\mbox{cut-set}}(\mathcal{T}).
\end{align*}
\end{theorem}

Precisely how large $t$ is required to be in order to approach $M_{\mbox{cut-set}}(\mathcal{T})$ within a specified tolerance is not clear for general networks.  However, an immediate consequence of Theorem \ref{thm:URcliqueRandom} is that $t=n-1$ is sufficient to achieve this lower bound with high probability when packets are randomly distributed in a fully connected network.

Finally, we remark that it is a simple exercise to construct examples where the cut-set bounds alone are not sufficient to characterize transmission schedules permitting universal recovery when packets are not divisible (e.g., a 4-node line network with packets $p_1$ and $p_2$ at the left-most and right-most nodes, respectively).  Thus, $t$-divisibility of packets provides the additional degrees of freedom necessary to approach the cut-set bounds more closely.  

\subsection{Remarks}
One interesting consequence of our results is that splitting packets does not significantly reduce the required number of packet-transmissions for many scenarios.  Indeed, at most one transmission can be saved if the network is fully connected (under any distribution of packets).  If the network is $d$-regular, we can expect to save fewer than $n$ transmissions if packets are randomly distributed (in fact, at most one transmission per node).  It seems possible that this result could be strengthened to include arbitrary distributions of packets in $d$-regular networks (as opposed to randomly distributed packets), but a proof  has not been found.

The limited value of dividing packets has practical ramifications since there is usually some additional communication overhead associated with dividing packets (e.g. additional headers, etc. for each transmitted chunk are required).  Thus, if the packets are very large, say each packet is a video file, our results imply that entire coded packets can be transmitted without significant loss, avoiding any additional overhead incurred by dividing packets. 

\section{An Application: Secrecy Generation} \label{sec:URSecrecy}
In this section, we  consider the setup of the cooperative data exchange problem for a fully connected network $\mathcal{T}$, but we consider a different goal.  In particular, we wish to generate a secret-key among the nodes that cannot be derived by an eavesdropper privy to all of the transmissions among nodes.  Also, like the nodes themselves, the eavesdropper is assumed to know the indices of the packets initially available to each node.  The goal is to generate the maximum amount of ``secrecy" that cannot be determined by the eavesdropper.  

The theory behind secrecy generation among multiple terminals was originally established in \cite{bib:CsiszarNarayanIT2004} for a very general class of problems. Our results should be interpreted as a practical application of the theory originally developed in \cite{bib:CsiszarNarayanIT2004}.  Indeed, our results and proofs are special cases of those in \cite{bib:CsiszarNarayanIT2004} which have been  streamlined to deal with the scenario under consideration.  The aim of the present section is to show how secrecy can be generated in a \emph{practical} scenario.  In particular, we show that it is possible to efficiently generate the  maximum amount of secrecy (as established in \cite{bib:CsiszarNarayanIT2004} ) among nodes in a fully connected network $\mathcal{T}=\{\mathcal{G},P_1,\dots,P_n\}$.  Moreover, we show that this is possible in the non-asymptotic regime (i.e., there are no $\epsilon$'s and we don't require the number of packets or nodes to grow arbitrarily large).  Finally, we note that it is possible to generate perfect secrecy instead of $\epsilon$-secrecy without any sacrifice.

\subsection{Practical Secrecy Results}
In this subsection, we state two results on secrecy generation.  Proofs are again postponed until Section \ref{sec:URProofs}. We begin with some definitions\footnote{We attempt to follow the notation of \cite{bib:CsiszarNarayanIT2004} where appropriate.}.  Let $\mathbf{F}$ denote the set of all transmissions (all of which are available to the eavesdropper by definition).  A function $K$ of the packets $\{p_1,\dots,p_k\}$ in the network is called a secret key (SK) if $K$ is recoverable by all nodes after observing $\mathbf{F}$, and it satisfies the (perfect) secrecy condition
\begin{align}
I(K;\mathbf{F})=0, \label{eqn:secrecy}
\end{align}
and the uniformity condition
\begin{align}
\Pr\left( K=key\right) = \frac{1}{|\mathcal{K}|} \mbox{~for all~}key\in\mathcal{K},
\end{align}
where $\mathcal{K}$ is the alphabet of possible keys.  

We define $C_{SK}(P_1,\dots,P_n)$ to be the secret-key capacity for a particular distribution of packets.  We will drop the notational dependence on $P_1, \dots,P_n$ where it doesn't cause confusion. By this we mean that a secret-key $K$ can be generated if and only if $\mathcal{K}=\mathbb{F}^{C_{SK}}$.  In other words, the nodes can generate at most $C_{SK}$ packets worth of secret-key.
Our first result of this section is the following:

\begin{theorem}\label{thm:SK}
The secret-key capacity is given by: $C_{SK}(P_1,\dots,P_n)=k-M^*(\mathcal{T})$.
\end{theorem}

Next, consider the related problem where a subset $D\subset V$ of nodes is compromised.  In this problem, the eavesdropper has access to $\mathbf{F}$ and $P_i$ for $i\in D$.  In this case, the secret-key should also be kept hidden from the nodes in $D$ (or else the eavesdropper could also recover it).  Thus, for a subset of nodes $D$, let $P_D=\bigcup_{i\in D} P_i$, and call $K$ a private-key (PK) if it is a secret-key which is only recoverable by the nodes in $V\backslash D$, and also satisfies the stronger secrecy condition:
\begin{align}
I(K;\mathbf{F},P_D)=0. \label{eqn:secrecyPK}
\end{align}
Similar to above, define $C_{PK}(P_1,\dots,P_n,D)$ to be the private-key capacity for a particular distribution of packets and subset of nodes $D$.  Again, we mean that a private-key $K$ can be generated if and only if $\mathcal{K}=\mathbb{F}^{C_{PK}}$.  In other words, the nodes in $V\backslash D$ can generate at most $C_{PK}$ packets worth of private-key.
Note that, since $P_D$ is known to the eavesdropper, each node $i\in D$ can transmit its respective set of packets $P_i$ without any loss of secrecy capacity.

Define a new network $\mathcal{T}_D=\{\mathcal{G}_D,\{P^{(D)}_i \}_{i\in V\backslash D}\}$ as follows.  Let $\mathcal{G}_D$ be the complete graph on $V\backslash D$, and let 
$P^{(D)}_i = P_i \backslash P_D$ for each $i\in V \backslash D$.  Thus, $\mathcal{T}_D$ is a fully connected network with $n-|D|$ nodes and $k-|P_D|$ packets.
Our second result of this section is the following:

\begin{theorem}\label{thm:PK}
The private-key capacity is given by: 
\begin{align*}
C_{PK}(P_1,\dots,P_n,D)=(k-|P_D|)-M^*(\mathcal{T}_D).
\end{align*}
\end{theorem}
The basic idea for private-key generation is that the users in $V\backslash D$ should  generate a secret-key from $\{p_1,\dots,p_k\} \backslash P_D$.

By the definitions of the SK and PK capacities, Theorem \ref{thm:URFullyConnectedPolyTime} implies that it is possible to compute these capacities efficiently.  Moreover, as we will see in the achievability proofs, these capacities can be achieved by performing coded cooperative data exchange amongst the nodes.  Thus, the algorithm developed in Appendix \ref{sec:Algo} combined with the algorithm in \cite{bib:JaggiIT2005} can be employed to efficiently solve the secrecy generation problem we consider. 

We conclude this subsection with an example to illustrate the results.

\begin{example}
Consider again the network of Example \ref{ex:B_N_clique} and assume $\mathbb{F}=\{0,1\}$ (i.e., each packet is a single bit).  The secret-key capacity for this network is $1$ bit.  After performing universal recovery, the eavesdropper knows $p_2$ and the parity $p_1\oplus p_3$.  A perfect secret-key is $K=p_1$ (we could alternatively use $K=p_3$).  If any of the nodes are compromised by the eavesdropper, the private-key capacity is 0.
\end{example}

We remark that the secret-key in the above example can in fact be attained by all nodes using only one transmission (i.e., universal recovery is not a prerequisite for secret-key generation).  However, it remains true that only one bit of secrecy can be generated.

\section{Proofs of Main Results} \label{sec:URProofs}
\subsection{Necessary and Sufficient Conditions for Universal Recovery}
\begin{proof}[Proof of Theorem \ref{thm:generalFeasibleRegion}]
This proof is accomplished by reducing the problem at hand to an instance of a single-source network coding problem and invoking the Max-Flow Min-Cut Theorem for network information flow \cite{bib:Ahlswede00networkinformation}.

First, fix the number of communication rounds $r$ to be large enough to permit universal recovery.  For a network $\mathcal{T}$, construct the network-coding graph $\mathcal{G}^{NC}=(V_{NC},E_{NC})$ as follows.  
The vertex set, $V_{NC}$ is defined as:
\begin{align*}
V_{NC}=\{s,u_1,\dots,u_k\} \cup \bigcup_{j=0}^{r} \{ v_1^j,\dots,v_n^j\} \cup \bigcup_{j=1}^{r}\{ w_1^j,\dots,w_n^j\}.
\end{align*}
The edge set, $E_{NC}$, consists of directed edges and is constructed as follows:
\begin{itemize}
\item For each $i\in [k]$, there is an edge of unit capacity\footnote{An edge of unit capacity can carry one field element $z\in\mathbb{F}$ per unit time.} from $s$ to $u_i$.
\item If $p_i\in {P}_j$, then there is an edge of infinite 
capacity from $u_i$ to $v_j^0$.
\item For each $j \in [r]$ and each $i\in [n]$, there is an edge of infinite capacity from $v_i^{j-1}$ to $v_i^{j}$.
\item For each $j \in [r]$ and each $i\in [n]$, there is an edge of capacity $b_i^{j}$ from $v_i^{j-1}$ to $w_i^{j}$.
\item For each $j \in [r]$ and each $i\in [n]$, there is an edge of infinite capacity from $w_i^{j}$ to $v_{i'}^{j}$ iff $i' \in \Gamma(i)$.
\end{itemize}

\begin{figure}
\centering
\def\svgwidth{2.5in}
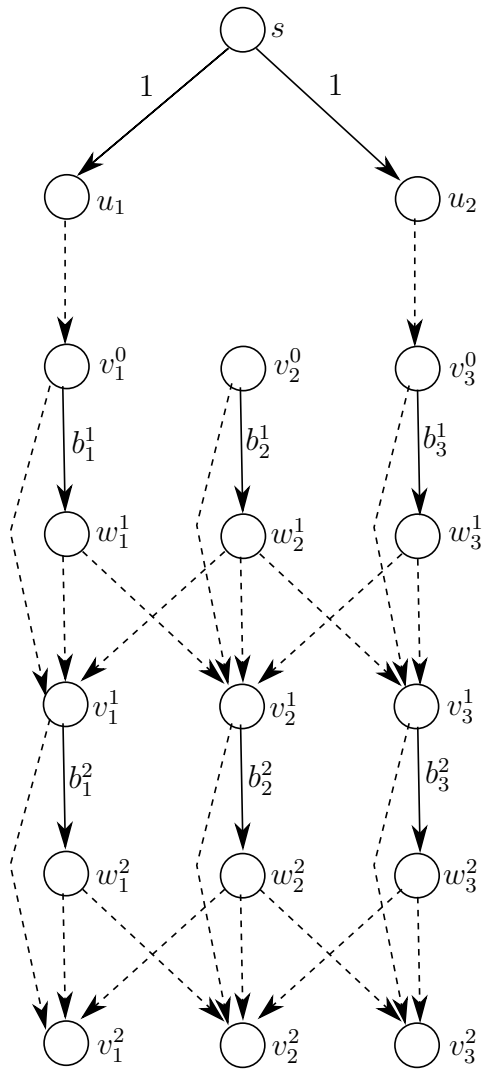
\caption{The graph $\mathcal{G}^{NC}$ corresponding to the line network of Example \ref{ex:B_N_line}.  Edges represented by broken lines have infinite capacity. Edges with finite capacities are labeled with the corresponding capacity value.}
\label{fig:G_NC}
\end{figure}

The interpretation of this graph is as follows: the vertex $u_i$ is introduced to represent packet $p_i$, the vertex $v_i^j$ represents node $i$ after the $j^{th}$ round of communication, and the vertex $w_i^j$ represents the broadcast of node $i$ during the $j^{th}$ round of communication.  If the $b_i^{j}$'s are chosen such that the graph $\mathcal{G}^{NC}$ admits a network coding solution which supports a multicast of $k$ units from $s$ to $\{v_1^r,\dots,v_n^r\}$, then this network coding solution also solves the universal recovery problem for the network $\mathcal{T}$ when node $i$ is allowed to make at most $b_i^{j}$ transmissions during the $j^{th}$ round of communication.  The graph $\mathcal{G}^{NC}$ corresponding to the line network of Example \ref{ex:B_N_line} is given in Figure \ref{fig:G_NC}.

We now formally prove the equivalence of the network coding problem on $\mathcal{G}^{NC}$ and the universal recovery problem defined by $\mathcal{T}$.

Suppose a set of encoding functions $\{f_i^j\}$ and a set decoding functions $\{\phi_i\}$ describe a transmission strategy which solves the universal recovery problem for a network $\mathcal{T}$ in $r$ rounds of communication.  Let $b_i^j$ be the number of transmissions made by node $i$ during the $j^{th}$ round of communication, and let $\mathcal{I}_i^j$ be all the information known to node $i$ prior to the $j^{th}$ round of communication (e.g. $\mathcal{I}_i^1={P}_i$).  The function $f_i^j$ is the encoding function for user $i$ during the $j^{th}$ round of communication (i.e. $f_i^j(\mathcal{I}_i^j) \in \mathbb{F}^{b_i^j}$), and the decoding functions satisfy:
\begin{align*}
\phi_i\left(\mathcal{I}_i^r,\cup_{i'\in\Gamma(i)} \{f_{i'}^r(\mathcal{I}_{i'}^r)\}\right)=\{p_1,\dots,p_k\}.
\end{align*}

Note that, given the encoding functions and the ${P}_i$'s, the $\mathcal{I}_i^j$'s can be defined recursively as:
\begin{align*}
\mathcal{I}_i^{j+1}=\mathcal{I}_i^{j} \cup \bigcup_{i'\in\Gamma(i)} \{f_{i'}^j(\mathcal{I}_{i'}^j)\}.
\end{align*}

The functions $\{f_i^j\}$ and $\{\phi_i\}$ can be used to generate a network coding solution which supports $k$ units of flow from $s$ to $\{v_1^r,\dots,v_n^r\}$ on  $\mathcal{G}^{NC}$ as follows:
  
For each vertex $v \in V_{NC}$, let $\mbox{IN}(v)$ be whatever $v$ receives on its incoming edges.  Let $g_v$ be the encoding function at vertex $v$, and $g_v(e,\mbox{IN}(v))$ be the encoded message which vertex $v$ sends along $e$ ($e$ is an outgoing edge from $v$).

If $e$ is an edge of infinite capacity emanating from $v$, let $g_v(e,\mbox{IN}(v))=\mbox{IN}(v)$.  

Let $s$ send $p_i$ along edge $(s,u_i)$.  At this point, we have $\mbox{IN}(v_i^0)={P}_i=\mathcal{I}_i^1$.  For each $i\in [n]$, let $g_{v_i^0}((v_i^0,w_i^1),\mbox{IN}(v_i^0))=f_i^1(\mathcal{I}_i^1)$.  By a simple inductive argument, defining the encoding functions $g_{v_i^j}((v_i^j,w_i^{j+1}),\mbox{IN}(v_i^j))$ to be equal to $f_i^{j+1}$ yields the result that $\mbox{IN}(v_i^r)=\left(\mathcal{I}_i^r,\cup_{i'\in\Gamma(i)} \{f_{i'}^r(\mathcal{I}_{i'}^r)\}\right)$.  Hence, the decoding function $\phi_i$ can be used at $v_i^r$ to allow error-free reconstruction of the $k$-unit flow.

The equivalence argument is completed by showing that a network coding solution which supports a $k$-unit multicast flow from $s$ to $\{v_1^r,\dots,v_n^r\}$ on  $\mathcal{G}^{NC}$ also solves the universal recovery problem on $\mathcal{T}$.  This is argued in a similar manner as above, and is therefore omitted.

Since we have shown that the universal recovery problem on $\mathcal{T}$ is equivalent to a network coding problem on $\mathcal{G}^{NC}$, the celebrated max-flow min-cut result of Ahlswede et. al \cite{bib:Ahlswede00networkinformation} is applicable.  In particular, a fixed vector $\{b_i^j\}$ admits a solution to the universal recovery problem where node $i$ makes at most $b_i^j$ transmissions during the $j^{th}$ round of communication if and only if any cut separating $s$ from some $v_i^r$ in $\mathcal{G}^{NC}$ has capacity at least $k$.

What remains to be shown is that the inequalities defining $\mathcal{R}_r(\mathcal{T})$ are satisfied if and only if any cut separating $s$ from some $v_i^r$ in $\mathcal{G}^{NC}$ has capacity at least $k$.

To this end, suppose we have a cut $({S},{S}^c)$ satisfying $s\in{S}^c$ and $v_i^r \in {S}$ for some $i\in [n]$. We will modify the cut $({S},{S}^c)$ to produce a new cut $({S}',{S}'^c)$ with capacity less than or equal to the capacity of the original cut $({S},{S}^c)$.

Define the set ${S}_0 \subseteq [n]$ as follows: $i\in {S}_0$ iff $v_i^r \in {S}$ (by definition of ${S}$, we have that ${S}_0\neq \emptyset$).

Initially, let ${S}'={S}$.  Modify the cut $({S}',{S}'^c)$ as follows:
\begin{enumerate}
\renewcommand{\labelenumi}{M\arabic{enumi})}
\item If $i\in \Gamma({S}_0)$, then place $w_i^r$ into ${S}'$. \label{mod:M1}
\item If $i \notin \Gamma({S}_0)$, then place $w_i^r$ into ${S}'^c$. \label{mod:M2}
\newcounter{enumM_saved}
\setcounter{enumM_saved}{\value{enumi}}
\end{enumerate}
Modifications M\ref{mod:M1} and M\ref{mod:M2} are justified (respectively) by J\ref{just:J1} and J\ref{just:J2}:
\begin{enumerate}
\renewcommand{\labelenumi}{J\arabic{enumi})}
\item If $i\in \Gamma({S}_0)$, then there exists an edge of infinite capacity from $w_i^r$ to some $v_{i'}^r\in{S}$.  Thus, moving $w_i^r$ to $\mathcal{S}'$ (if necessary) does not increase the capacity of the cut.\label{just:J1}
\item If $i \notin \Gamma({S}_0)$, then there are no edges from $w_i^r$ to ${S}$, hence we can move $w_i^r$ into ${S}'^c$ (if necessary) without increasing the capacity of the cut.\label{just:J2}
\newcounter{enumJ_saved}
\setcounter{enumJ_saved}{\value{enumi}}
\end{enumerate}

Modifications M\ref{mod:M1} and M\ref{mod:M2} guarantee that $w_i^r\in {S}'$ iff $i\in\Gamma({S}_0)$.  Thus, assume that $({S}',{S}'^c)$ satisfies this condition and further modify the cut as follows:
\begin{enumerate}
\renewcommand{\labelenumi}{M\arabic{enumi})}
\setcounter{enumi}{\value{enumM_saved}}
\item If $i\in {S}_0$, then place $v_i^{r-1}$ into ${S}'$.\label{mod:M3}
\item If $i \notin \Gamma({S}_0)$, then place $v_i^{r-1}$ into ${S}'^c$.\label{mod:M4}
\setcounter{enumM_saved}{\value{enumi}}
\end{enumerate}
Modifications M\ref{mod:M3} and M\ref{mod:M4} are justified (respectively) by J\ref{just:J3} and J\ref{just:J4}:
\begin{enumerate}
\renewcommand{\labelenumi}{J\arabic{enumi})}
\setcounter{enumi}{\value{enumJ_saved}}
\item If $i\in {S}_0$, then there exists an edge of infinite capacity from $v_i^{r-1}$ to $v_{i}^r\in{S}$.  Thus, moving $v_i^{r-1}$ to $\mathcal{S}'$ (if necessary) does not increase the capacity of the cut.\label{just:J3}
\item If $i \notin \Gamma({S}_0)$, then there are no edges from $v_i^{r-1}$ to ${S}'$ (since $w_i^r\notin {S}'$ by assumption), hence we can move $v_i^{r-1}$ into ${S}'^c$ (if necessary) without increasing the capacity of the cut.\label{just:J4}
\setcounter{enumJ_saved}{\value{enumi}}
\end{enumerate}

At this point, define the set ${S}_1 \subseteq [n]$ as follows: $i\in {S}_1$ iff $v_i^{r-1} \in {S}'$.  Note that the modifications of ${S}'$ guarantee that ${S}_1$ satisfies ${S}_0 \subseteq {S}_1 \subseteq \Gamma({S}_0)$.

This procedure can be repeated for each layer of the graph resulting in a sequence of sets $\emptyset \subsetneq {S}_0 \subseteq \dots \subseteq {S}_r \subseteq [n]$ satisfying ${S}_{j} \subseteq \Gamma({S}_{j-1})$ for each $j\in [r]$.

We now perform a final modification of the cut $({S}',{S}'^c)$:
\begin{enumerate}
\renewcommand{\labelenumi}{M\arabic{enumi})}
\setcounter{enumi}{\value{enumM_saved}}
\item If $p_j \in \cup_{i\in{S}_r} {P}_i$, then place $u_j$ into ${S}'$. \label{mod:M5}
\item If $p_j \notin \cup_{i\in{S}_r} {P}_i$, then place $u_j$ into ${S}'^c$. \label{mod:M6}
\setcounter{enumM_saved}{\value{enumi}}
\end{enumerate}
Modifications M\ref{mod:M5} and M\ref{mod:M6} are justified (respectively) by J\ref{just:J5} and J\ref{just:J6}:
\begin{enumerate}
\renewcommand{\labelenumi}{J\arabic{enumi})}
\setcounter{enumi}{\value{enumJ_saved}}
\item If $p_j \in \cup_{i\in{S}_r} {P}_i$, then there is an edge of infinite capacity from $u_j$ to ${S}'$ and moving $u_j$ into ${S}'$ (if necessary) does not increase the capacity of the cut. \label{just:J5}
\item If $p_j \notin \cup_{i\in{S}_r} {P}_i$, then there are no edges from $u_j$ to ${S}'$, hence moving $u_j$ (if necessary) into ${S}'^c$ cannot increase the capacity of the cut. \label{just:J6}
\setcounter{enumM_saved}{\value{enumi}}
\end{enumerate}
\renewcommand{\labelenumi}{\arabic{enumi})}

A quick calculation shows that the modified cut $({S}',{S}'^c)$ has capacity greater than or equal to $k$ iff:
\begin{align}
\sum_{j=1}^r & \sum_{i\in {S}_j^c \cap \Gamma({S}_{j-1})} b^{r+1-j}_i \geq \left| \bigcap_{i\in {S}_r} {P}_i^c \right|. \label{eqn:NecSufConst}
\end{align}

Since every modification of the cut either preserved or reduced the capacity of the cut, the original cut $({S},{S}^c)$ also has capacity greater than or equal to $k$ if the above inequality is satisfied.  In Figure \ref{fig:G_NC_cut}, we illustrate a cut $(S,S^c)$ and its modified minimal cut $(S',S'^c)$ for the graph $\mathcal{G}^{NC}$ corresponding to the line network of Example \ref{ex:B_N_line}.

\begin{figure}
\centering
\def\svgwidth{3.5in}
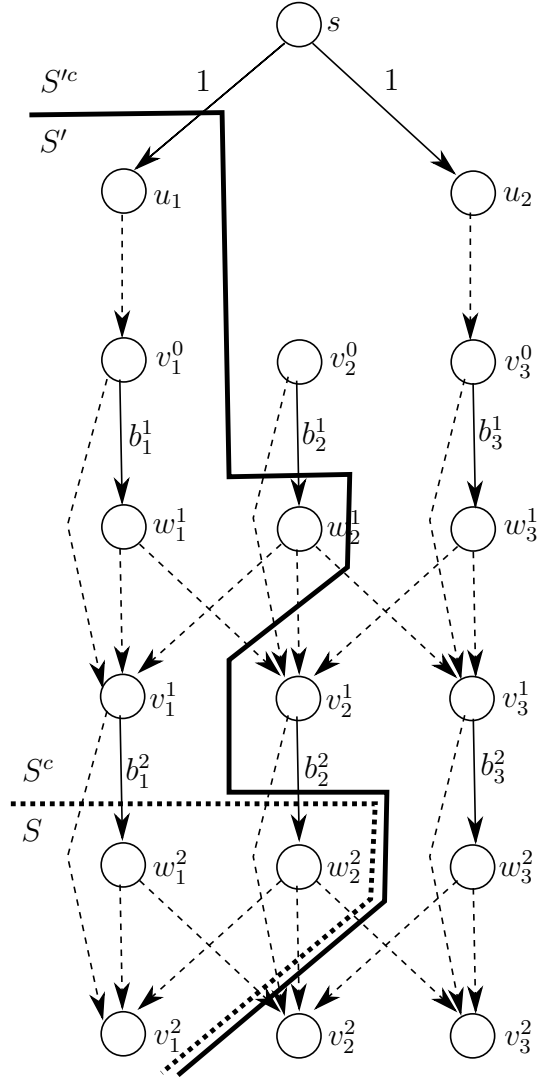
\caption{The graph $\mathcal{G}^{NC}$ corresponding to the line network of Example \ref{ex:B_N_line} with original cut $(S,S^c)$ and the corresponding modified minimal cut $(S',S'^c)$.  In this case, $S_0=S_1=S_2=\{1\}$.  Upon substitution into \eqref{eqn:NecSufConst}, this choice of $S_0,S_1,S_2$ yields the inequality $b_2^1+b_2^2 \geq 1$.}
\label{fig:G_NC_cut}
\end{figure}

By the equivalence of the universal recovery problem on a network $\mathcal{T}$ to the network coding problem on $\mathcal{G}^{NC}$ and the max-flow min-cut theorem for network information flow, if a transmission scheme solves the universal recovery problem on $\mathcal{T}$, then the associated $b_i^j$'s must satisfy the constraints of the form given by (\ref{eqn:NecSufConst}).  Conversely, for any set of $b_i^j$'s which satisfy the constraints of the form given by (\ref{eqn:NecSufConst}), there exists a transmission scheme using exactly those numbers of transmissions which solves the universal recovery problem for $\mathcal{T}$.  Thus the constraints of (\ref{eqn:NecSufConst}),  and hence the inequalities defining $\mathcal{R}_r(\mathcal{T})$, are satisfied if and only if any cut separating $s$ from some $v_i^r$ in $\mathcal{G}^{NC}$ has capacity at least $k$.

\begin{remark}
Since $\left| \bigcap_{i\in [n]} {P}_i^c \right|=0$, constraints where ${S}_r=[n]$ are trivially satisfied.  Therefore, we can restrict our attention to sequences of sets where ${S}_r\subsetneq [n]$.
\end{remark}
\end{proof}

\subsection{Fully Connected Networks}
\begin{proof}[Proof of Theorem \ref{thm:URFullyConnectedPolyTime}]
In the case where $\mathcal{T}$ is a fully connected network, we have that ${S}_j^c\cap\Gamma({S}_{j-1})={S}_j^c$ for any  nonempty $S \subset V$.  Therefore, the constraints defining $\mathcal{R}_r(\mathcal{T})$ become:
\begin{align}
\sum_{j=1}^r  & \sum_{i\in {S}_j^c}  b^{r+1-j}_i \geq \left| \bigcap_{i\in {S}_r} {P}_i^c \right|. \label{eqn:CliqueNecSuff}
\end{align}
Now, suppose a transmission schedule $\{b_i^j\}\in \mathcal{R}_r(\mathcal{T})$ and consider the modified transmission schedule  $\{\tilde{b}_i^j\}$ defined by: $\tilde{b}_i^r=\sum_{j=1}^r b_i^j$ and $\tilde{b}_i^j=0$ for $j<r$.  By construction,  ${S}_{j+1}^c\subseteq {S}_{j}^c$ in the constraints defining $\mathcal{R}_r(\mathcal{T})$. Therefore, using the definition of $\{\tilde{b}_i^j\}$, we have:
\begin{align*}
\sum_{i\in {S}_1^c}  \tilde{b}^{r}_i
\geq \sum_{j=1}^r \sum_{i\in {S}_j^c}  b^{r+1-j}_i 
\geq \left| \bigcap_{i\in {S}_r} {P}_i^c \right|.
\end{align*}

Thus the modified transmission schedule is also in $\mathcal{R}_r(\mathcal{T})$.  Since $\left| \bigcap_{i\in {S}_1} {P}_i^c \right| \geq \left| \bigcap_{i\in {S}_r} {P}_i^c \right|$, when $\mathcal{T}$ is a fully connected network, it is sufficient to consider constraints of the form:
\begin{align}
\sum_{i\in {S}^c} b_i^1 &\geq  \left| \bigcap_{i\in {S}} {P}_i^c \right| \mbox{~~for all nonempty $S\subset V$}.\label{eqn:CliqueCutset}
\end{align}
This proves the latter two statements of the theorem: that the cut-set constraints are necessary and sufficient for universal recovery when $\mathcal{T}$ is a fully connected network, and that a single round of communication is sufficient to achieve universal recovery with $M^*(\mathcal{T})$ transmissions.

With these results established, an optimal transmission schedule can be obtained by solving the following integer linear program:
\begin{align}
\mbox{minimize~~} &\sum_{i =1}^n b_i \label{ilp:cliqueNet}\\
\mbox{subject to:~~} & \sum_{i\in S^c} b_i \geq \left| \bigcap_{i\in S} P_i^c \right| \mbox{~~for each nonempty $S \subset V$}.\notag
\end{align} 

In order to accomplish this, we identify $B_i \leftarrow P_i^c$  and set $w_i=1$ for $i\in [n]$ and apply the submodular algorithm presented in Appendix \ref{sec:Algo}.
\end{proof}

Now we consider fully connected networks in which packets are randomly distributed according to \eqref{eqn:probModel}, which is parametrized by $q$.  The proof of Theorem \ref{thm:URcliqueRandom} requires the following lemma:

\begin{lemma}\label{lem:convexity}
If $0<q<1$ is fixed, then there exists some $\delta>0$ such that the following inequality holds for all $\ell\in\{2,\dots,n-1\}$:
\begin{align}
\frac{n-\ell}{n-1}\geq \frac{(1-q)^\ell-(1-q)^n}{1-q-(1-q)^n}+\delta.\notag
\end{align}
\end{lemma}
\begin{proof}
Applying Jensen's inequality to the strictly convex function $f(x)=(1-q)^x$ using the convex combination $\ell=\theta\cdot 1+(1-\theta)\cdot n$ yields:
\begin{align}
\frac{(1-q)^\ell-(1-q)^n}{1-q-(1-q)^n} <\frac{n-\ell}{n-1}. \notag
\end{align}
Taking $\delta$ to be the minimum gap in the above inequality for the values $\ell\in\{2,\dots,n-1\}$ completes the proof.
\end{proof}

\begin{proof}[Proof of Theorem \ref{thm:URcliqueRandom}]
We begin by showing that the LP
\begin{align}
\mbox{minimize~~} &\sum_{i =1}^n b_i \label{eqn:LPforCliqueProof}\\
\mbox{subject to:~~} & \sum_{i\in S^c} b_i \geq \left| \bigcap_{i\in S} P_i^c \right| \mbox{~~for each nonempty $S \subset V$}. \label{eqn:LPineqConstRandom}
\end{align} 
has an optimal value of $\frac{1}{n-1}\sum_{i=1}^n |P_i^c|$ with high probability.  To this end, note that the inequalities 
\begin{align}
\sum_{\substack{i=1\\i\neq j}}^n b_i \geq |P_j^c|  \mbox{~~for $1\leq j \leq n$.} \label{eqn:subsetofLPConstraints}
\end{align}
are a subset of the inequality constraints \eqref{eqn:LPineqConstRandom}. Summing both sides of \eqref{eqn:subsetofLPConstraints} over $1\leq j\leq n$ reveals that any feasible vector $b\in \mathbb{R}^n$ for LP \eqref{eqn:LPforCliqueProof}-\eqref{eqn:LPineqConstRandom} must satisfy:
\begin{align}
\sum_{i=1}^n b_i \geq \frac{1}{n-1}\sum_{i=1}^n |P_i^c|. \label{eqn:CutSetLowerBound}
\end{align}
This establishes a lower bound on the optimal value of the LP.  We now identify a solution that is feasible with probability approaching 1 as $k\rightarrow \infty$ while achieving the lower bound of \eqref{eqn:CutSetLowerBound} with equality. To begin note that 
\begin{align}
\tilde{b}_j = \frac{1}{n-1} \sum_{i=1}^n  |P_i^c| -  |P_j^c| \label{eqn:b_iSoln}
\end{align}
is a solution to the system of linear equations given by \eqref{eqn:subsetofLPConstraints} and achieves \eqref{eqn:CutSetLowerBound} with equality.  Now, we prove that $(\tilde{b}_1,\dots,\tilde{b}_n)$ is a feasible solution to LP \eqref{eqn:LPforCliqueProof} with high probability.  To be specific, we must verify that
\begin{align}
\sum_{i \in {S}^c } \tilde{b}_i \geq \left|\bigcap_{i\in {S}} {P}_i^c\right| \label{eqn:InequalitiesforOtherSubsets}
\end{align}
holds with high probability for all subsets ${S}\subset V$ satisfying $2\leq |{S}| \leq n-1$ (the case $|{S}|=1$ is satisfied by the definition of $\{\tilde{b}_i\}_{i=1}^n$).  Substitution of (\ref{eqn:b_iSoln}) into (\ref{eqn:InequalitiesforOtherSubsets}) along with some algebra yields that the following equivalent conditions must hold:
\begin{align}
\left(\frac{n-|{S}|}{n-1} \right) \sum_{i=1}^n \frac{1}{k}  \left|{P}_i^c\right| -
\sum_{i\in {S}^c} \frac{1}{k} \left|{P}^c_i\right| \geq \frac{1}{k} \left| \bigcap_{i\in {S}}{P}_i^c\right|. \label{eqn:equivCondn}
\end{align}

To this end, note that for any ${S}$, $\left| \bigcap_{i\in {S}}{P}_i^c\right|$ is a random variable which can be expressed as $\left| \bigcap_{i\in {S}}{P}_i^c\right|=\sum_{j=1}^k X_j^{{S}}$, where  $X_j^{{S}}$ is an indicator random variable taking the value $1$ if $p_j\in \bigcap_{i\in {S}}{P}_i^c $ and $0$ otherwise.  From \eqref{eqn:probModel} we have:
\begin{align*}
\Pr\left(X_j^{{S}}=1\right)=\frac{(1-q)^{|{S}|}-(1-q)^n}{1-(1-q)^n}.
\end{align*}

By the weak law of large numbers, for any $\eta>0$:
\begin{align}
\Pr\left(\left| 
\frac{1}{k} \bigg| \bigcap_{i\in {S}}{P}_i^c\bigg|-
\frac{(1-q)^{|{S}|}-(1-q)^n}{1-(1-q)^n}
\right| > \eta \right) &<\epsilon_k, \label{eqn:WLLN}
\end{align}
where $\epsilon_k\rightarrow 0$ as $k\rightarrow\infty$.
Thus, by the union bound, Lemma \ref{lem:convexity}, and taking $\eta$ sufficiently small, the following string of inequalities 
holds with arbitrarily high probability as $k\rightarrow\infty$:
\begin{align*}
&\left(\frac{n-|{S}|}{n-1} \right) \sum_{i=1}^n \frac{1}{k}  \left|{P}_i^c\right| -
\sum_{i\in {S}^c} \frac{1}{k} \left|{P}^c_i\right| \\
&\geq  
\left(\frac{n-|{S}|}{n-1} \right)
\left(\frac{(1-q)-(1-q)^n}{1-(1-q)^n} -(2n-1)\eta\right) \\
&\geq\frac{(1-q)^{{|S|}}-(1-q)^n}{1-(1-q)^n}+\eta\\
&\geq 
\frac{1}{k} \left| \bigcap_{i\in {S}}{P}_i^c\right|.
\end{align*}
These steps are justified as follows: for $\eta$ sufficiently small the first and last inequalities hold with high probability by \eqref{eqn:WLLN}, and the second inequality follows from Lemma \ref{lem:convexity} with $\ell=|S|$.
This proves that \eqref{eqn:equivCondn} holds, and therefore $(\tilde{b}_1,\dots,\tilde{b}_n)$ is a feasible solution to LP \eqref{eqn:LPforCliqueProof} with high probability.  Now, taking Corollary \ref{cor:ILP_LPgap} in Appendix \ref{sec:Algo} together with Theorem \ref{thm:URFullyConnectedPolyTime} completes the proof.
\end{proof}

\subsection{$d$-Regular Networks}
\begin{lemma} \label{lem:dregRandomLP}
Assume packets are randomly distributed in a $d$-regular network $\mathcal{T}$.  For any $\epsilon>0$, there exists an optimal solution $x^*$ to LP (\ref{lp:dregObj}-\ref{lp:dregConst}) which satisfies
\begin{align*}
 \left\|x^*-\frac{1}{d}\mathbb{E}[|P_1^c|] \mathds{1} \right\|_{\infty} < \epsilon k 
 \end{align*}
with probability approaching $1$ as $k\rightarrow \infty$, where $\mathbb{E}$ indicates expectation.
\end{lemma}
\begin{proof}
Let $\vec{P}=(|P_1^c|,\dots,|P_n^c|)^T$ and let $A$ be the adjacency matrix of $\mathcal{G}$ (i.e., $a_{i,j}=1$ if $(i,j)\in E$ and $0$ otherwise).  Observe that $A$ is symmetric and $A\mathds{1}=d\mathds{1}$, where $\mathds{1}$ denotes a column vector of $1$'s.  With this notation, LP \eqref{lp:dregObj} can be rewritten as: 
\begin{align}
\mbox{minimize~~} &\mathds{1}^T x \label{lp:dregObj2}\\
\mbox{subject to:~~} & Ax \succeq \vec{P}, \notag
\end{align} 
where ``$a\succeq b$" for vectors $a,b\in \mathbb{R}^n$ means that $a_i\geq b_i$ for $i=1,\dots,n$. 

Let $A^+$ denote the Moore-Penrose pseudoinverse of $A$.  Observe that the linear least squares solution to $Ax\approx \vec{P}$ is given by:
\begin{align*}
\bar{x}_{LS} &= A^{+}\vec{P} \\
&= A^{+}\mathbb{E}\vec{P}\ + A^{+}\left(\vec{P}-\mathbb{E}\vec{P}\right)\\
&=\frac{1}{d}\mathbb{E}\vec{P} + A^{+}\left(\vec{P}-\mathbb{E}\vec{P}\right).
\end{align*}
For the last step above, note that $\mathbb{E}\vec{P}$ is an eigenvector of $A$ with eigenvalue $d$ so $\mathbb{E}\vec{P}$ will also be an eigenvector of $A^+$ with eigenvalue $\frac{1}{d}$.  Hence,
\begin{align*}
\|x_{LS}-\frac{1}{d}\mathbb{E}\vec{P} \|_2 &= \|A^{+}\left(\vec{P}-\mathbb{E}\vec{P}\right)\|_2 \\
&\leq \|A^{+}\|_2 \|\vec{P}-\mathbb{E}\vec{P}\|_2. 
\end{align*}
Combining this with the triangle inequality implies that, for any vector $y$,
\begin{align*}
\|y-\frac{1}{d}\mathbb{E}\vec{P}\|_{\infty} &\leq \|y - \bar{x}_{LS} \|_{\infty}+ \|\bar{x}_{LS} - \frac{1}{d}\mathbb{E}\vec{P} \|_{\infty}\\
&\leq \|y - \bar{x}_{LS} \|_{\infty}+ \|\bar{x}_{LS} - \frac{1}{d}\mathbb{E}\vec{P} \|_2\\
&\leq \|y - \bar{x}_{LS} \|_{\infty}+ \|A^{+}\|_2 \|\vec{P}-\mathbb{E}\vec{P}\|_2.
\end{align*}

Therefore, Lemma \ref{lem:LP_LSapprox} (see Appendix \ref{app:LPlemma}) guarantees the  existence of an optimal solution $x^*$ to LP \eqref{lp:dregObj2} (and consequently LP \eqref{lp:dregObj}) which satisfies:
\begin{align*}
\|x^*-\frac{1}{d}\mathbb{E}\vec{P}\|_{\infty} &\leq \|x^* - \bar{x}_{LS} \|_{\infty}+ \|A^{+}\|_2 \|\vec{P}-\mathbb{E}\vec{P}\|_2\\
&\leq c_A \|A\bar{x}_{LS} - \vec{P} \|_2+ \|A^{+}\|_2 \|\vec{P}-\mathbb{E}\vec{P}\|_2\\
&\leq c_A \|\frac{1}{d}A\mathbb{E}\vec{P}  - \vec{P} \|_2 + \|A^{+}\|_2 \|\vec{P}-\mathbb{E}\vec{P}\|_2\\
&=c_A \| \mathbb{E}\vec{P}  - \vec{P} \|_2+ \|A^{+}\|_2 \|\vec{P}-\mathbb{E}\vec{P}\|_2,
\end{align*}
where $c_A$ is a constant depending only on $A$. By the weak law of large numbers, $\| \vec{P}-\mathbb{E}\vec{P} \|_{2} \leq \epsilon k$ with probability tending to $1$ as $k\rightarrow \infty$ for any $\epsilon>0$. Noting that $\mathbb{E}\vec{P}=\mathbb{E}[|P_1^c|] \mathds{1}$ completes the proof.
\end{proof}

\begin{proof}[Proof of Theorem \ref{thm:URdReg}]
We begin with some observations and definitions:  
\begin{itemize}
\item First, recall that our model for randomly distributed packets \eqref{eqn:probModel} implies that
\begin{align}
\mathbb{E}\left[\left|\bigcap_{i\in S}P_i^c\right|\right] = k\frac{(1-q)^{|S|}-(1-q)^n}{1-(1-q)^n} \mbox{~~for all nonempty $S\subset V$.}
\end{align}
\item With this in mind, there exists a constant $c_q>0$ such that 
\begin{align}
\mathbb{E}\left[\left| P_1^c\right|\right] \geq (1+c_q)\mathbb{E}\left[\left|\bigcap_{i\in S}P_i^c\right|\right] \mbox{~~for all $S\subset V, |S| \geq 2$.} \label{eqn:cqIneq}
\end{align}
\item Next, Lemma \ref{lem:convexity} implies the existence of a constant $\delta_q>0$ such that for any $S\subset V$ with $2\leq |S|\leq n-1$:
\begin{align}
\frac{n-|S|}{n-1} \geq \frac{(1-q)^{|S|}-(1-q)^n}{(1-q)-(1-q)^n} +\delta_q = \frac{\mathbb{E}\left[\left|\bigcap_{i\in S}P_i^c\right|\right]}{\mathbb{E}\left[\left| P_1^c\right|\right] } + \delta_q. \label{eqn:cvxLemExp}
\end{align}
\item The weak law of large numbers implies that
\begin{align}
\left(1+\frac{\min \{\delta_q, c_q \}}{4} \right) \mathbb{E}\left[\left|\bigcap_{i\in S}P_i^c\right|\right] \geq \left|\bigcap_{i\in S}P_i^c\right|
\label{eqn:expectationgtactual}
\end{align}
with probability approaching $1$ as $k\rightarrow \infty$.
\item Finally, for the proof below, we will take the number of communication rounds sufficiently large to satisfy
\begin{align}
r\geq \max\left\{\frac{2d}{n \delta_q}, \frac{2n(1+c_q)}{d c_q} \right\}. \label{eqn:defnR}
\end{align}
\end{itemize}
Fix $\epsilon>0$. Lemma \ref{lem:dregRandomLP} guarantees that there exists an optimal solution $x^*$ to LP \eqref{lp:dregObj} satisfying 
\begin{align}
 \left\|x^*-\frac{1}{d}\mathbb{E}[|P_1^c|] \mathds{1} \right\|_{\infty} < \epsilon k \label{eqn:xinftybounds}
\end{align}
with probability tending to $1$ in $k$.  Now, it is always possible to construct a transmission schedule $\{b_i^j\}$ which satisfies 
$\sum_j b_i^j = \lceil x_i^* \rceil$ and $\lfloor \frac{1}{r}x_i^* \rfloor \leq b_i^j \leq \lceil\frac{1}{r}x_i^*\rceil$ for each $i,j$.   Observe that $\sum_{i,j} b_i^j < n+ \sum_{i} x_i^*$.  Thus, proving that $\{b_i^j\}\in \mathcal{R}_r(\mathcal{T})$ with high probability will prove the theorem.

Since the network is $d$-regular, $|\partial({S})|\geq d$ whenever $|{S}|\leq n-d$ and $|\partial({S}_1)| \geq n-|{S}_2|$ whenever $|{S}_2|\geq n-d$ and ${S}_1\subseteq {S}_2$.  We consider the cases where $2\leq|{S}_r|\leq n-d$ and $n-d < |{S}_r|\leq n-1$ separately.  The case where $ |{S}_r|=1$ coincides precisely with the constraints \eqref{lp:dregConst}, and hence is satisfied by definition of $\{b_i^j\}$.

Considering the case where $2\leq|{S}_r|\leq n-d$, we have the following string of inequalities:
\begin{align}
\sum_{j=1}^r \sum_{i\in{S}_j^c\cap\Gamma({S}_{j-1})} b_i^{(r+1-j)}  
&\geq \sum_{j=1}^r \sum_{i\in{S}_j^c\cap\Gamma({S}_{j-1})} \left\lfloor \frac{1}{r}x_i^* \right\rfloor  \label{eqn:defnBis} \\
&\geq \frac{1}{r}\sum_{j=1}^r \sum_{i\in{S}_j^c\cap\Gamma({S}_{j-1})} x_i^* -nr \label{eqn:ineqSeries2}\\
&= \frac{1}{r}\sum_{j=1}^r \sum_{i\in\partial({S}_{j-1})} x_i^* - \frac{1}{r}\sum_{i\in {S}_r \cap {S}_0^c} x_i^*-nr \label{eqn:ineqSeries3}\\
&\geq \frac{1}{r}\sum_{j=1}^r \sum_{i\in\partial({S}_{j-1})} \frac{1}{d}\mathbb{E}[|P_1^c|] -
\frac{1}{r}\sum_{i\in {S}_r \cap {S}_0^c} \frac{1}{d}\mathbb{E}[|P_1^c|]  - nk \epsilon  -nr \label{eqn:ineqSeries4} \\
&\geq \frac{1}{rd}\mathbb{E}[|P_1^c|]  \left(\sum_{j=1}^r\left|\partial({S}_{j-1})\right| -n \right) - nr(k \epsilon+1) \label{eqn:ineqSeries5}  \\
&\geq \frac{1+c_q}{rd}\mathbb{E}\left[\left|\bigcap_{i\in S_r}P_i^c\right|\right]  \left(\sum_{j=1}^r\left|\partial({S}_{j-1})\right| -n \right) - nr(k \epsilon+1) \label{eqn:ineqSeries6} \\
&\geq \frac{1+c_q}{rd}\mathbb{E}\left[\left|\bigcap_{i\in S_r}P_i^c\right|\right]  \left(rd -n \right) - nr(k \epsilon+1) \label{eqn:ineqSeries7}\\
&\geq \left(1+\frac{c_q}{2}\right) \mathbb{E}\left[\left|\bigcap_{i\in S_r}P_i^c\right|\right]   - nr(k \epsilon+1) \label{eqn:ineqSeries8} \\
&\geq \left(1+\frac{c_q}{4}\right) \mathbb{E}\left[\left|\bigcap_{i\in S_r}P_i^c\right|\right] \label{eqn:ineqSeries9} \\
&\geq \left|\bigcap_{i\in S_r}P_i^c\right|. \label{eqn:ineqSeries10}
\end{align}
The above string of inequalities holds with probability tending to $1$ as $k\rightarrow \infty$.  They can be justified as follows:
\begin{itemize}
\item \eqref{eqn:defnBis} follows by  definition of $\{b_i^j\}$.
\item \eqref{eqn:ineqSeries2} follows since $\left\lfloor \frac{1}{r}x_i^* \right\rfloor \geq \frac{1}{r}x_i^*-1$ and $|S_j^c \cap \Gamma(S_{j-1})|\leq n$.
\item \eqref{eqn:ineqSeries3} follows from writing $\cup_{j=1}^{r}{S}_j^c\cap\Gamma({S}_{j-1})$ as $\left( \cup_{j=1}^{r}\partial({S}_{j-1}) \right) \backslash \left( S_r \cap {S}_{0}^{c} \right )$ and 
 expanding the sum.
\item \eqref{eqn:ineqSeries4} follows from \eqref{eqn:xinftybounds}.
\item \eqref{eqn:ineqSeries5}  is true since $|S_0^c \cap S_r|\leq n$.
\item \eqref{eqn:ineqSeries6} follows from \eqref{eqn:cqIneq}.
\item \eqref{eqn:ineqSeries7} follows from $\left|\partial({S}_{j-1})\right|\geq d$ by $d$ regularity and the assumption that $2\leq|{S}_r|\leq n-d$.
\item \eqref{eqn:ineqSeries8} follows from  our choice of $r$ given in \eqref{eqn:defnR}.
\item \eqref{eqn:ineqSeries9} follows since $\frac{c_q}{4}\mathbb{E}\left[\left|\bigcap_{i\in S_r}P_i^c\right|\right]  \geq nr(k \epsilon+1)$ with high probability
for $\epsilon$ sufficiently small.
\item \eqref{eqn:ineqSeries10} follows from  \eqref{eqn:expectationgtactual}.
\end{itemize}

Next, consider the case where $n-d\leq |{S}_r|\leq n-1$. Starting from \eqref{eqn:ineqSeries5}, we obtain:
\begin{align}
\sum_{j=1}^r \sum_{i\in{S}_j^c\cap\Gamma({S}_{j-1})} b_i^{(r+1-j)}  
&\geq \frac{1}{rd}\mathbb{E}[|P_1^c|]  \left(\sum_{j=1}^r\left|\partial({S}_{j-1})\right| -n \right) - nr(k \epsilon+1) \label{eqn:ineqSeries21}\\
&\geq \frac{1}{rd}\mathbb{E}[|P_1^c|]  \bigg(r(n-|S_r|) -n \bigg) - nr(k \epsilon+1) \label{eqn:ineqSeries22}\\
&=\mathbb{E}[|P_1^c|]  \left(\frac{n-|S_r|}{d} -\frac{n}{rd} \right) - nr(k \epsilon+1) \label{eqn:ineqSeries23}\\
&\geq\mathbb{E}[|P_1^c|]  \left(\frac{n-|S_r|}{n-1} -\frac{n}{rd} \right) - nr(k \epsilon+1) \label{eqn:ineqSeries24}\\
&\geq\mathbb{E}[|P_1^c|]  \left(\frac{\mathbb{E}\left[\left|\bigcap_{i\in S_r}P_i^c\right|\right] }{\mathbb{E}[|P_1^c|] } +\delta_q-\frac{n}{rd} \right) - nr(k \epsilon+1) \label{eqn:ineqSeries25}\\
&\geq\mathbb{E}[|P_1^c|]  \left(\frac{\mathbb{E}\left[\left|\bigcap_{i\in S_r}P_i^c\right|\right] }{\mathbb{E}[|P_1^c|] } +\frac{\delta_q}{2} \right) - nr(k \epsilon+1) \label{eqn:ineqSeries26}\\
&\geq \left(1+\frac{\delta_q}{4}\right) \mathbb{E}\left[\left|\bigcap_{i\in S_r}P_i^c\right|\right] \label{eqn:ineqSeries27}\\
&\geq \left|\bigcap_{i\in S_r}P_i^c\right|. \label{eqn:ineqSeries28}
\end{align}

The above string of inequalities holds with probability tending to $1$ as $k\rightarrow \infty$.  They can be justified as follows:
\begin{itemize}
\item \eqref{eqn:ineqSeries21} is simply \eqref{eqn:ineqSeries5} repeated for convenience.
\item \eqref{eqn:ineqSeries22} follows since $n-d\leq |{S}_r|\leq n-1$ and hence $d$-regularity implies that $|\partial(S_{j-1})|\geq(n-|S_r|)$.
\item \eqref{eqn:ineqSeries24} follows since $d\leq n-1$.
\item \eqref{eqn:ineqSeries25} follows from \eqref{eqn:cvxLemExp}.
\item \eqref{eqn:ineqSeries26} follows from from our definition of $r$ given in \eqref{eqn:defnR}.
\item \eqref{eqn:ineqSeries27} follows since $\frac{\delta_q}{4} \mathbb{E}[P_1^c]  \geq nr(k \epsilon+1)$ with high probability
for $\epsilon$ sufficiently small.
\item \eqref{eqn:ineqSeries28} follows from \eqref{eqn:expectationgtactual}.
\end{itemize}

Thus, we conclude that, for $\epsilon$ sufficiently small, the transmission schedule $\{b_i^j\}$ satisfies each of the inequalities defining $\mathcal{R}_r(\mathcal{T})$ with probability tending to $1$.  Since the number of such inequalities is finite, an application of the union bound completes the proof that $\{b_i^j\}\in\mathcal{R}_r(\mathcal{T})$ with probability tending to $1$ as $k\rightarrow \infty$.
\end{proof}

\subsection{Divisible Packets}
\begin{proof}[Proof of Theorem \ref{thm:URtsplit}]
Fix any $\epsilon>0$ and let $x^*$ be an optimal solution to LP \eqref{lp:dvisObj}.  Put $b_i=x_i^*+\epsilon$.  Note that $b_i$ is nonnegative.  This follows by considering the set $S\backslash \{i\}$ in the inequality constraint \eqref{lp:divisConst}, which implies $x_i^*\geq 0$.

Now, take an integer $r \geq \epsilon^{-1}n \max_{1\leq i\leq n} b_i$. If packets are $t$-divisible, we can find a transmission schedule $\{b_i^j\}$ such that $\frac{1}{r} b_i \leq b_i^j \leq \frac{1}{r} b_i + \frac{1}{t}$ for all $i\in[n], j\in[r]$.  

Thus, for any $({S}_0,\cdots,{S}_r) \in \mathcal{S}^{(r)}(\mathcal{G})$ we have the following string of inequalities:
\begin{align*}
\sum_{j=1}^r \sum_{i\in S_j^c \cap \Gamma(S_{j-1})} b_i^{(r+1-j)} &\geq \frac{1}{r} \sum_{j=1}^r \sum_{i\in S_j^c \cap \Gamma(S_{j-1})} b_i \\
&=\frac{1}{r} \sum_{j=1}^r \sum_{i\in \partial(S_{j-1}) } b_i - \frac{1}{r} \sum_{i\in S_0^c\cap S_r } b_i\\
&=\frac{1}{r} \sum_{j=1}^r \sum_{i\in \partial(S_{j-1}) } x_i^* + \frac{\epsilon}{r} \sum_{j=1}^r |\partial(S_{j-1}) |  - \frac{1}{r} \sum_{i\in S_0^c\cap S_r } b_i \\
&\geq \frac{1}{r} \sum_{j=1}^r \left| \bigcap_{i\in S_{j-1}} P_j^c \right| + \epsilon  - \frac{n}{r} \max_{1\leq i \leq n} b_i  \\
&\geq \left| \bigcap_{i\in S_{r}} P_j^c \right|.
\end{align*}
Hence, Theorem \ref{thm:generalFeasibleRegion} implies that the transmission schedule $\{b_i^j\}$ is sufficient to achieve universal recovery.  Noting that 
\begin{align*}
\sum_{i,j} b_i^j \leq \sum_{i=1}^n b_i + \frac{n r}{t} \leq \sum_{i=1}^n x^*_i +n\left(\frac{r}{t}+\epsilon\right)
\end{align*}
completes the proof of the theorem.
\end{proof}

\subsection{Secrecy Generation}
In this subsection, we prove Theorems \ref{thm:SK} and \ref{thm:PK}.  We again remark that our proofs can be seen as special cases of those in \cite{bib:CsiszarNarayanIT2004} which have been adapted for the problem at hand.  For notational convenience, define $P=\{p_1,\dots,p_k\}$.
We will require the following lemma.  

\begin{lemma}\label{lem:SKPK}
Given a packet distribution $P_1,\dots,P_n$, let $K$ be a secret-key achievable with communication $\mathbf{F}$.  Then the following holds:
\begin{align}
H(K|\mathbf{F})=H(P)-\sum_{i=1}^n x_i.
\end{align}
for some vector $x=(x_1,\dots,x_n)$ which is feasible for the following ILP:
\begin{align}
\mbox{minimize~~} &\sum_{i =1}^n x_i \label{eqn:ILPforSecrecyProof}\\
\mbox{subject to:~~} & \sum_{i\in S} x_i \geq \left| \bigcap_{i\in S^c} P_i^c \right| \mbox{~~for each nonempty $S \subset V$}. \label{eqn:ILPineqConstSecrecy}
\end{align}

Moreover, if $K$ is a PK (with respect to a set $D$) and each node $i\in D$ transmits its respective set of packets $P_i$, then
\begin{align}
H(K|\mathbf{F})=H(P|P_D)-\sum_{i\in V\backslash D} x_i.
\end{align}
for some vector $x=(x_1,\dots,x_n)$ which is feasible for the ILP:
\begin{align}
\mbox{minimize~~} &\sum_{i \in V\backslash D} x_i \label{eqn:ILPforSecrecyProofD}\\
\mbox{subject to:~~} & \sum_{i\in S} x_i \geq \left| \bigcap_{i\in S^c} P_i  ^c \right| \mbox{~~for each nonempty $S \subset V\backslash D$}. \label{eqn:ILPineqConstSecrecyD}
\end{align} 
\end{lemma}

\begin{remark}
We remark that \eqref{eqn:ILPineqConstSecrecy} and \eqref{eqn:ILPineqConstSecrecyD} are necessary and sufficient conditions for achieving universal recovery in the networks $\mathcal{T}$ and $\mathcal{T}_D$ considered in Theorems \ref{thm:SK} and \ref{thm:PK}, respectively.  Thus, the optimal values of ILPs \eqref{eqn:ILPforSecrecyProof} and \eqref{eqn:ILPforSecrecyProofD} are equal to $M^*(\mathcal{T})$ and $M^*(\mathcal{T}_D)$, respectively.
\end{remark}

\begin{proof}
We assume throughout that all entropies are with respect to the base-$|\mathbb{F}|$ logarithm (i.e., information is measured in packets). For this and the following proofs, let $\mathbf{F}=(F_1,\dots,F_n)$ and $F_{[1,i]}=(F_1,\dots,F_i)$, where $F_i$ denotes the transmissions made by node $i$.  For simplicity, our proof does not take into account interactive communication, but can be modified to do so.  Allowing interactive communication does not change the results.  See \cite{bib:CsiszarNarayanIT2004} for details.  

Since $K$ and $\mathbf{F}$ are functions of $P$:
\begin{align}
H(P)&=H(\mathbf{F},K,P_1,\dots ,P_n)\\
&=\sum_{i=1}^n H(F_i|F_{[1,i-1]})+H(K|\mathbf{F})+\sum_{i=1}^n H(P_i|\mathbf{F},K,P_{[1,i-1]}).
\end{align}
Set $x_i=H(F_i|F_{[1,i-1]})+H(P_i|\mathbf{F},K,P_{[1,i-1]})$.  Then, the substituting $x_i$ into the above equation yields:
\begin{align}
H(K|\mathbf{F})=H(P)-\sum_{i=1}^n x_i.
\end{align}
To show that $x=(x_1,\dots, x_n)$ is a feasible vector for ILP \eqref{eqn:ILPforSecrecyProof}, we write:
\begin{align}
\left| \bigcap_{i\in  S^c} P_i^c \right| &= H(P_S|P_{S^c})\\
&=H(\mathbf{F},K,P_S|P_{S^c})\\
&=\sum_{i=1}^n H(F_i|F_{[1,i-1]},P_{S^c})+H(K|\mathbf{F},P_{S^c})+\sum_{i\in S}H(P_i|\mathbf{F},K,P_{[1,i-1]},P_{S^c\cap[i+1,n]})\\
&\leq \sum_{i\in S} H(F_i|F_{[1,i-1]}) + \sum_{i\in S} H(P_i|\mathbf{F},K,P_{[1,i-1]})\\
&=\sum_{i\in S} x_i.
\end{align}
In the above inequality, we used the fact that conditioning reduces entropy, the fact that $K$ is a function of $(\mathbf{F},P_{S^c})$ for any $S\neq V$, and the fact that $F_i$ is a function of $P_i$ (by the assumption that communication is not interactive).

To prove the second part of the lemma, we can assume $D=\{1,\dots ,\ell\}$.  The assumption that each node $i$ in $D$ transmits all of the packets in $P_i$ implies $F_i=P_i$.  Thus, for $i\in D$ we have $x_i=H(P_i|P_{[1,i-1]})$.  Repeating the above argument, we obtain
\begin{align}
H(K|\mathbf{F})&=H(P)-H(P_{D})-\sum_{i\in V\backslash D} x_i\\
&=H(P|P_{D})-\sum_{i\in V \backslash D} x_i,
\end{align}
completing the proof of the lemma.
\end{proof}

\begin{proof}[Proof of Theorem \ref{thm:SK}]
\emph{Converse Part}.  Suppose $K$ is a secret-key achievable with communication $\mathbf{F}$.  Then, by definition of a SK and Lemma \ref{lem:SKPK} we have
\begin{align}
C_{SK}=H(K)=H(K|\mathbf{F})= H(P)-\sum_{i=1}^n x_i \leq H(P)-M^*(\mathcal{T})=k-M^*(\mathcal{T}).
\end{align}

\emph{Achievability Part}.  By definition, universal recovery can be achieved with $M^*(\mathcal{T})$ transmissions.  Moreover, the communication $\mathbf{F}$ can be generated as a linear function of $P$ (see the proof of Theorem \ref{thm:generalFeasibleRegion} and \cite{bib:JaggiIT2005}).  Denote this linear transformation by $\mathbf{F}=\mathcal{L} P$. Note that $\mathcal{L}$ only depends on the indices of the packets available to each node, not the values of the packets themselves (see \cite{bib:JaggiIT2005}). Let $\mathcal{P}_{\mathbf{F}}=\{P':\mathcal{L}P'=\mathbf{F}\}$ be the set of all packet distributions which generate $\mathbf{F}$.  

By our assumption that the packets are i.i.d.~uniform from $\mathbb{F}$, each $P'\in\mathcal{P}_{\mathbf{F}}$ is equally likely given $\mathbf{F}$ was observed.  Since $\mathbf{F}$ has dimension $M^*(\mathcal{T})$, $|\mathcal{P}_{\mathbf{F}}|=\mathbb{F}^{k-M^*(\mathcal{T})}$.  Thus, we can set $\mathcal{K}=\mathbb{F}^{k-M^*(\mathcal{T})}$ and label each $P'\in\mathcal{P}_{\mathbf{F}}$ with a unique element in $\mathcal{K}$.  The label for the actual $P$ (which is reconstructed by all nodes after observing $\mathbf{F}$) is the secret-key.  Thus, $C_{SK}\geq k-M^*(\mathcal{T})$.  

We remark that this labeling can be done efficiently by an appropriate linear transformation mapping $P$ to $K$.
\end{proof}

\begin{proof}[Proof of Theorem \ref{thm:PK}]
\emph{Converse Part}. Suppose $K$ is a private-key.  Then, by definition of a PK and Lemma \ref{lem:SKPK},
\begin{align*}
C_{PK}=H(K)=H(K|\mathbf{F})&= H(P|P_D)-\sum_{i\in V\backslash D} x_i \\
&\leq H(P|P_D)-M^*(\mathcal{T}_D)=(k-|P_D|)-M^*(\mathcal{T}_D).
\end{align*}

\vspace{12pt}
\emph{Achievability Part}.  Let each node $i\in D$ transmit $P_i$ so that we can update $P_j\leftarrow P_j \cup P_D$ for each $j\in V \backslash D$.  Now, consider the universal recovery problem for only the nodes in $V \backslash D$.  $M^*(\mathcal{T}_D)$ is the minimum number of transmissions required among the nodes in $V\backslash D$ so that each node in $V\backslash D$ recovers $P$.  At this point, the achievability proof proceeds identically to the SK case.
\end{proof}

\section{Concluding Remarks} \label{sec:URConclusion}
In this paper, we derive necessary and sufficient conditions for achieving universal recovery in an arbitrarily connected network.  For the case when the network is fully connected, we provide an efficient algorithm based on submodular optimization which efficiently solves the cooperative problem.  This algorithm and its derivation yield tight concentration results for the case when packets are randomly distributed.  Moreover, concentration results are provided when the network is $d$-regular and packets are distributed randomly.  If packets are divisible, we prove that the traditional cut-set bounds are achievable.  As a consequence of this and the  concentration results, we show that splitting packets does not typically provide a significant benefit when the network is $d$-regular.  Finally, we discuss an application to secrecy generation in the presence of an eavesdropper.  We demonstrate that our submodular algorithm can be used to generate the maximum amount of secrecy in an efficient manner.

It is conceivable that the coded cooperative data exchange problem can be solved (or approximated) in polynomial time if the network is $d$-regular, but packets aren't necessarily randomly distributed.  This is one possible direction for future work.

\section*{Acknowledgement}
The authors would like to thank Kent Benson, Alex Sprintson, Pavan Datta, and Chi-Wen Su for the helpful conversations and suggestions which led to this paper.

\appendices

\section{An Efficiently Solvable Integer Linear Program}\label{sec:Algo}
In this appendix, we introduce a special ILP and provide an efficient algorithm for solving it. This algorithm can be used to efficiently solve the  cooperative data exchange problem when the underlying graph is fully-connected. We begin by introducing some notation\footnote{We attempt to keep the notation generic in order to emphasize that the results in this appendix are not restricted to the context of the {cooperative data exchange} problem.}.

Let $E=\{1,\dots,n\}$ be a finite set with $n$ elements.  We denote the family of all subsets of $E$ by $2^E$.  We frequently use the compact notation $E\backslash U$ and $U+i$ to denote the sets $E\cap U^c$ and $U\cup\{i\}$ respectively.  For a vector $x=(x_1,\dots,x_n)\in \mathbb{R}^n$, define the corresponding functional $x:2^E\rightarrow \mathbb{R}$ as: \begin{align}
x(U):=\sum_{i\in U} x_i, \mbox{~for~} U\subseteq E.
\end{align}

Throughout this section, we let $\mathcal{F} = 2^E -\{\emptyset,E\}$ denote the family of nonempty proper subsets of $E$.  Let $\mathcal{B}=\{B_1,\dots,B_n\}$.  No special structure is assumed for the $B_i$'s except that they are finite.

With the above notation established, we consider the following Integer Linear Program (ILP) in this section:
\begin{align}
\mbox{minimize}\left\{\sum_{i\in E} w_i x_i : x(U) \geq \left| \bigcap_{i\in E \backslash U} B_i \right| , \forall ~U\in\mathcal{F} , x_i \in \mathbb{Z} \right\}.\label{eqn:generalILP}
\end{align}
It is clear that any algorithm that efficiently solves this ILP also solves ILP \eqref{ilp:cliqueNet} by putting $B_i\leftarrow P_i^c$ and $w=\mathds{1}$.

\subsection{Submodular Optimization}
Our algorithm for solving ILP \eqref{eqn:generalILP} relies heavily on submodular function optimization.  To this end, we give a very brief introduction to submodular functions here.

A function $g:2^E\rightarrow \mathbb{R}$ is said to be submodular if, for all $X,Y\in 2^E$,
\begin{align}
g(X)+g(Y)\geq g(X \cap Y)+g(X \cup Y).
\end{align}
Over the past three decades, submodular function optimization has received a significant amount of attention.
Notably, several polynomial time algorithms have been developed for solving  the Submodular Function Minimization (SFM) problem
\begin{align}
\min\left\{g(U):U\subseteq E \right\}.
\end{align}
We refer the reader to \cite{bib:McCormick2005,bib:Fujishige2010,bib:Schrijver2003} for a comprehensive overview of SFM and known algorithms.  As we will demonstrate, we can solve ILP \eqref{eqn:generalILP} via an algorithm that iteratively calls a SFM routine.  The most notable feature of SFM algorithms is their ability to solve problems with exponentially many constraints in polynomial time.  One of the key drawbacks of SFM is that the problem formulation is very specific.  Namely, SFM routines typically require the function $g$ to be submodular on \emph{all} subsets of the set E.

\subsection{The Algorithm}
We begin by developing an algorithm to solve an equality constrained version of ILP \eqref{eqn:generalILP}. We will remark on the general case at the conclusion of this section.  To this end, let $M$ be a positive integer and consider the following ILP:
\begin{align}
\mbox{minimize~~~~} w^Tx& \label{eqn:ILP}\\
\mbox{subject to:~} x(U) &\geq \left| \bigcap_{i\in E \backslash U} B_i \right| \mbox{~for all~} U\in\mathcal{F},\mbox{~and} \label{eqn:ILPeq1}\\
x(E) &= M.\label{eqn:ILPeq2}
\end{align}

\begin{remark}We assume $w_i\geq 0$, else in the case without the equality constraint we could allow the corresponding $x_i\rightarrow +\infty$ and the problem is unbounded from below.
\end{remark}

\begin{pseudocode}[framebox]{SolveILP}{\mathcal{B},E,M,w}\label{alg:TopLevel}
\COMMENT{Define $f:2^E\rightarrow \mathbb{R}$ as in equation \eqref{eqn:defineF}.}\\
x \leftarrow \CALL{ComputePotentialX}{f,M,w}\\
\IF \CALL{CheckFeasible}{f,x}
\THEN \RETURN{x}
\ELSE \RETURN{\mbox{Problem Infeasible}}
\end{pseudocode}

\begin{theorem}
Algorithm \ref{alg:TopLevel} solves the equality constrained ILP \eqref{eqn:ILP} in polynomial time.  If feasible, Algorithm \ref{alg:TopLevel} returns an optimal $x$.  If infeasible, Algorithm \ref{alg:TopLevel} returns ``Problem Infeasible".
\end{theorem}
\begin{proof}
The proof is accomplished in three steps:
\begin{enumerate}
\item First, we show that if our algorithm returns an $x$, it is feasible.
\item Second, we prove that if a returned $x$ is feasible, it is also optimal.
\item Finally, we show that if our algorithm does not return an $x$, then the problem is infeasible.
\end{enumerate}
Each step is given its own subsection.
\end{proof}

Algorithm \ref{alg:TopLevel} relies on three basic subroutines given below:

\begin{pseudocode}[framebox]{ComputePotentialX}{f,M,w} \label{alg:ComputeX}
\COMMENT{If feasible, returns $x$ satisfying \eqref{eqn:ILPeq1} and \eqref{eqn:ILPeq2} that minimizes $w^T x$.}\\
\COMMENT{Order elements of $E$ so that $w_{1}\geq w_{2} \geq \dots \geq w_{n}$.}\\

\FOR i\GETS n \TO 2 \DO
\BEGIN
\COMMENT{Define $f_{i}(U):=f(U+i)$ for $U\subseteq \{i,\dots,n\}$.}\\
x_i \GETS \CALL{SFM}{f_{i},\{i,\dots,n\}}\\
\END \\
x_1 \GETS M-\sum_{i=2}^n x_i\\
\RETURN{x}
\end{pseudocode}

\begin{pseudocode}[framebox]{CheckFeasible}{f,x}\label{alg:checkFeasible}
\COMMENT{Check if $x(U)\leq f(U)$ for all $U \in \mathcal{F}$ with $1 \in U$.}\\
\COMMENT{Define $f_{1}(U):=f(U+1)$ for $U\subseteq E$.}\\
\IF \CALL{SFM}{f_{1},E} < 0 ~ \THEN \RETURN \FALSE
\ELSE \RETURN \TRUE
\end{pseudocode}

\begin{pseudocode}[framebox]{SFM}{f,V}
\COMMENT{Minimize submodular function $f$ over groundset $V$. See \cite{bib:McCormick2005} for details.}\\
v \GETS  \min \left\{f(U) : U \subseteq V\right\}\\
\RETURN{v}
\end{pseudocode}

\subsection{Feasibility of a Returned $x$}
In this section, we prove that if Algorithm \ref{alg:TopLevel} returns a vector $x$, it must be feasible.  We begin with some definitions.

\begin{definition}
A pair of sets $X,Y \subset E$ is called \textbf{crossing} if $X\cap Y \neq \emptyset$ and $X\cup Y \neq E$.
\end{definition}

\begin{definition}
A function $g:2^E \rightarrow \mathbb{R}$ is \textbf{crossing submodular} if \[g(X)+g(Y)\geq g(X\cap Y)+g(X\cup Y)\] for $X,Y$ crossing.
\end{definition}

We remark that minimization of crossing submodular functions is well established, however it involves a lengthy reduction to a standard submodular optimization problem.  However, the crossing family $\mathcal{F}$ admits a straightforward algorithm, which is what we provide in Algorithm \ref{alg:TopLevel}.  We refer the reader to \cite{bib:Schrijver2003} for complete details on the general case.

For $M$ a positive integer, define
\begin{align}
f(U):=M-\left|\bigcap_{i\in U} B_i \right|-x(U),\mbox{~for~} U\in \mathcal{F}.\label{eqn:defineF}
\end{align}

\begin{lemma}
The function $f$ is crossing submodular on $\mathcal{F}$.
\end{lemma}
\begin{proof}
For $X,Y\in \mathcal{F}$ crossing:
\begin{align*}
f(X)+f(Y)&=M-\left|\bigcap_{i\in X} B_i \right|-x(X)+M-\left|\bigcap_{i\in Y} B_i \right|-x(Y)\\
&=M-\left|\bigcap_{i\in X} B_i \right|-x(X\cap Y)+M-\left|\bigcap_{i\in Y} B_i \right|-x(X\cup Y)\\
&\geq M-\left|\bigcap_{i\in X\cap Y} B_i \right|-x(X\cap Y)+M-\left|\bigcap_{i\in X \cup Y} B_i \right|-x(X\cup Y)\\
&=f(X \cap Y)+f(X \cup Y).
\end{align*}
\end{proof}

Observe that, with $f$ defined as above, the constraints of ILP \eqref{eqn:ILP} can be equivalently written as:
\begin{align}
f(U)=M-\left| \bigcap_{i\in U} B_i \right| -x(U)&\geq 0 \mbox{~for all~} U\in\mathcal{F},\mbox{~and} \label{eqn:ineqType}\\
x(E) &= M.
\end{align}

Without loss of generality, assume the elements of $E$ are ordered lexicographically so that $w_1 \geq w_2 \geq \dots \geq w_n$.  At iteration $i$ in Algorithm \ref{alg:ComputeX}, $x_j=0$ for all $j\leq i$.  Thus, setting
\begin{align}
x_i &\leftarrow \min_{U\subseteq \{i,\dots,n\}}\left\{ f_i(U) \right\}\\
&=\min_{U\subseteq \{i,\dots,n\}:i\in U}\left\{ f(U) \right\}\\
&=\min_{U\subseteq \{i,\dots,n\}:i\in U}\left\{ M-\left| \bigcap_{i\in U} B_i \right| -x(U) \right\} \label{eqn:Uk}
\end{align}
and noting that the returned $x$ satisfies $x(E)=M$, rearranging \eqref{eqn:Uk} guarantees that
\begin{align}
x(E\backslash U) \geq \left| \bigcap_{i\in U} B_i \right|, \mbox{~for all~} U\subseteq \{i,\dots,n\},i\in U \label{eqn:feas1}
\end{align}
as desired.  Iterating through $i\in\{2,\dots,n\}$ guarantees \eqref{eqn:feas1} holds for $2\leq i\leq n$.

\begin{remark}
In the feasibility check routine (Algorithm \ref{alg:checkFeasible}), we must be able to evaluate $f_1(E)$.  The reader can verify that putting $f(E)=0$ preserves submodularity.
\end{remark}

Now, in order for the feasibility check to return \textbf{true}, we must have
\begin{align}
\min_{U\subseteq E}\left\{ f_1(U) \right\}
&=\min_{U\subseteq E:1\in U}\left\{ f(U) \right\}\\
&=\min_{U\subseteq E:1\in U}\left\{ M-\left| \bigcap_{i\in U} B_i \right| -x(U) \right\}\\
&\geq 0,
\end{align}
implying that
\begin{align}
x(E\backslash U) \geq \left| \bigcap_{i\in U} B_i \right|, \mbox{~for all~} U\subseteq E, 1\in U. \label{eqn:feas2}
\end{align}
Combining \eqref{eqn:feas1} and \eqref{eqn:feas2} and noting that $x(E)=M$ proves that $x$ is indeed feasible.  Moreover, $x$ is integral as desired.

\subsection{Optimality of a Returned $x$}\label{sec:optim}
In this section, we prove that if Algorithm \ref{alg:TopLevel} returns a feasible $x$, then it is also optimal.  First, we require two more definitions and a lemma.

\begin{definition}
A constraint of the form \eqref{eqn:ineqType} corresponding to $U$ is said to be \textbf{tight} for $U$ if
\begin{align}
f(U) = M - \left| \bigcap_{i\in U} B_i \right| - x(U) = 0.
\end{align}
\end{definition}
\begin{lemma}\label{lem:tight}
If $x$ is feasible, $X,Y$ are crossing, and their corresponding constraints are tight, then the constraints corresponding to $X\cap Y$ and $X \cup Y$ are also tight.
\end{lemma}
\begin{proof}
Since the constraints corresponding to $X$ and $Y$ are tight, we have
\begin{align}
0=f(X)+f(Y) \geq f(X\cap Y)+f(X\cup Y) \geq 0.
\end{align}
The first inequality is due to submodularity and the last inequality holds since $x$ is feasible.  This implies the result.
\end{proof}

\begin{definition}
A family of sets $\mathcal{L}$ is \textbf{laminar} if $X,Y\in \mathcal{L}$ implies either $X\cap Y=\emptyset$, $X\subset Y$, or $Y\subset X$.
\end{definition}

At iteration $k$ ($1 < k \leq n$) of Algorithm \ref{alg:ComputeX}, let $U_k$ be the set where \eqref{eqn:Uk} achieves its minimum.  Note that $k\in U_k \subseteq \{k,\dots,n\}$.  By construction, the constraint corresponding to $U_k$ is tight.  Also, the constraint $x(E)=M$ is tight.  From the $U_k$'s and $E$ we can construct a laminar family as follows: if $U_j\cap U_k\neq \emptyset$ for $j<k$, then replace $U_j$ with  $\tilde{U}_j \leftarrow U_k \cup  U_j$.  By Lemma \ref{lem:tight}, the constraints corresponding to the sets in the newly constructed laminar family are tight. Call this family $\mathcal{L}$.  For each $i\in E$, there is a unique smallest set in $\mathcal{L}$ containing $i$.  Denote this set $L_i$.  Since $k\in U_k \subseteq \{k,\dots,n\}$, $L_i\neq L_j$ for $i\neq j$. Note that $L_1=E$ and $L_i \subset L_j$ only if $j<i$.

For each $L_i\in\mathcal{L}$ there is a unique smallest set $L_{j}$ such that $L_{i}\subset L_{j}$.  We call $L_{j}$ the least upper bound on $L_i$.

Now, consider the dual linear program to \eqref{eqn:ILP}:
\begin{align}
\mbox{maximize~~} &-\sum_{U\in\mathcal{F}} \pi_U \left(M-\left| \bigcap_{i\in U} B_i \right|\right) - \pi_E M \\
\mbox{subject to:~} &\sum_{U\in\mathcal{F}: i \in U} \pi_U +\pi_E + w_i = 0, \mbox{~for~}1\leq i \leq n\\
&\pi_U \geq 0 \mbox{~for~} U\in \mathcal{F}, \mbox{~and~} \pi_E \mbox{~free}.
\end{align}

For each $L_{i} \in \mathcal{L}$, let the corresponding dual variable $\pi_{L_{i}}=w_{j}-w_{i}$, where $L_j$ is the least upper bound on $L_i$.  By construction, $\pi_{L_{i}} \geq 0$ since it was assumed that $w_1\geq \dots \geq w_n$.  Finally, let $\pi_E = -w_1$ and $\pi_U=0$ for $U\notin \mathcal{L}$.

Now, observe that:
\begin{align}
\sum_{U\in\mathcal{F}: i \in U} \pi_U +\pi_E + w_i = 0
\end{align}
as desired for each $i$.  Thus, $\pi$ is dual feasible.  Finally, note that $\pi_U>0$ only if $U\in\mathcal{L}$.  However, the primal constraints corresponding to the sets in $\mathcal{L}$ are tight.  Thus, $(x,\pi)$ form a primal-dual feasible pair satisfying complementary slackness conditions, and are therefore optimal.

\subsection{No Returned $x$ = Infeasibility} \label{sec:feas}

Finally, we prove that if the feasibility check returns \textbf{false}, then ILP \eqref{eqn:ILP} is infeasible.  Note by construction that the vector $x$ passed to the feasibility check satisfies
\begin{align}
M - \left| \bigcap_{i\in U} B_i \right| - x(U) \geq 0 \mbox{~for all nonempty~} U\subseteq\{2,\dots,n\},
\end{align}
and $x(E)=M$. Again, let $U_k$ be the set where \eqref{eqn:Uk} achieves its minimum and let $\mathcal{L}$ be the laminar family generated by these $U_k$'s and $E$ exactly as before.  Again, the constraints corresponding to the sets in $\mathcal{L}$ are tight (this can be verified in a manner identical to the proof of Lemma \ref{lem:tight}).  Now, since $x$ failed the feasibilty check, there exists some exceptional set $T$ with $1\in T$ for which
\begin{align}
M - \left| \bigcap_{i\in T} B_i \right| - x(T) < 0.
\end{align}
Generate a set $L_T$ as follows: Initialize $L_T \leftarrow T$.  For each $L_i\in \mathcal{L}, L_i\neq E$, if $L_T \cap L_i \neq \emptyset$, update $L_T \leftarrow L_T\cup L_i$. Now, we can add $L_T$ to family $\mathcal{L}$ while preserving the laminar property.  We pause to make two observations:
\begin{enumerate}
\item By an argument similar to the proof of Lemma \ref{lem:tight}, we have that
\[M - \left| \bigcap_{i\in L_T} B_i \right| - x(L_T) < 0.\]
\item The sets in $\mathcal{L}$ whose least upper bound is $E$ form a partition of $E$.  We note that $L_T$ is a nonempty class of this partition.  Call this partition $\mathcal{P}_\mathcal{L}$.
\end{enumerate}
Again consider the dual constraints, however, let $w_i=0$ (this does not affect feasibility).  For each $L\in \mathcal{P}_\mathcal{L}$ define the associated dual variable $\pi_L = \alpha$, and let $\pi_E=-\alpha$.  All other dual variables are set to zero. It is easy to check that this $\pi$ is dual feasible.  Now, the dual objective function becomes:
\begin{align}
-\sum_{U\in\mathcal{F}} \pi_U \left((M-\left| \bigcap_{i\in U} B_i \right|\right) - \pi_E M
&=  -\alpha\sum_{L \in \mathcal{P}_{\mathcal{L}} } \left(M-\left| \bigcap_{i\in L} B_i \right| -x(L) + x(L) \right) +\alpha M\\
&= -\alpha \left(M-\left| \bigcap_{i\in  L_T} B_i \right| -x(L_T) \right) -\alpha x(E)+\alpha M\\
&= -\alpha \left(M-\left| \bigcap_{i\in  L_T} B_i \right| -x(L_T) \right)\\
&\rightarrow +\infty \mbox{~as~} \alpha\rightarrow \infty.
\end{align}
Thus, the dual is unbounded and therefore the primal problem must be infeasible.

As an immediate corollary we obtain the following:

\vspace{12pt}
\begin{corollary}\label{cor:ILP_LPgap}
The optimal values of the ILP: \[\min\left\{x(E): x(U)\geq \left| \cap_{i\in E\backslash U} B_i\right|, U\in\mathcal{F},x_i\in \mathbb{Z}\right\}\] and the corresponding LP relaxation: \[\min\left\{x(E): x(U)\geq \left| \cap_{i\in E\backslash U} B_i\right|, U\in\mathcal{F},x_i\in \mathbb{R}\right\}\] differ by less than 1.
\end{corollary}
\vspace{12pt}

\begin{proof}
Algorithm \ref{alg:TopLevel} is guaranteed to return an optimal $x$ if the intersection of the polytope and the hyperplane $x(E)=M$ is nonempty.  Thus, if $M^*$ is the minimum such $M$, then the optimal value of the LP must be greater than $M^*-1$.
\end{proof}

\subsection{Solving the General ILP}

Finally, we remark on how to solve the general case of the ILP without the equality constraint given in \eqref{eqn:generalILP}.  First, we state a simple convexity result.

\begin{lemma}
Let $p^*_w(M)$ denote the optimal value of ILP \eqref{eqn:ILP} when the equality constraint is $x(E)=M$.  We claim that $p^*_w(M)$ is a convex function of $M$.
\end{lemma}
\begin{proof}
Let $M_1$ and $M_2$ be integers and let $\theta\in [0,1]$ be such that $M_{\theta}=\theta M_1+(1-\theta)M_2$ is an integer. Let $x^{(1)}$ and be $x^{(2)}$ optimal vectors that attain $p^*_w(M_1)$ and $p^*_w(M_2)$ respectively.  Let $x^{(\theta)}=\theta x^{(1)}+(1-\theta) x^{(2)}$.  By convexity, $x^{(\theta)}$ is feasible, though not necessarily integer. However, by the results from above, optimality is always attained by an integral vector. Thus, it follows that:
\begin{align}
\theta p^*_w(M_1)+(1-\theta) p^*_w(M_2)=\theta w^T x^{(1)}+(1-\theta) w^T x^{(2)} =w^T x^{(\theta)}  \geq p^*_w(M_\theta).
\end{align}
\end{proof}

Noting that $p^*_w(M)$ is convex in $M$, we can perform bisection on $M$ to solve the ILP in the general case.  For our purposes, it suffices to have relatively loose upper and lower bounds on $M$ since the complexity only grows logarithmically in the difference.  A simple  lower bound on $M$ is  given by $M\geq \max_i \left| B_i \right|$.

\subsection{Complexity}
Our aim in this paper is not to give a detailed complexity analysis of our algorithm.  This is due to the fact that the complexity is dominated by the the SFM over the set $E$ in Algorithm \ref{alg:checkFeasible}.  Therefore, the complexity of Algorithm \ref{alg:TopLevel} is essentially the same as the  complexity of the SFM solver employed.

However, we have performed a series of numerical experiments to demonstrate that Algorithm \ref{alg:TopLevel} performs quite well in practice.  In our implementation, we ran the Fujishige-Wolfe (FW) algorithm for SFM \cite{bib:Fujishige06theMinNormPoint} based largely on a Matlab routine by A. Krause \cite{bib:Krause_sfo}.  While the FW algorithm has not been proven to run in polynomial time, it has been shown to work quite well in practice \cite{bib:Fujishige06theMinNormPoint} (similar to the Simplex algorithm for solving Linear Programs).  Whether or not FW has worst-case polynomial complexity is an open problem to date.  We remark that there are several SFM algorithms that run in strongly polynomial time which could be used if a particular application requires polynomially bounded worst-case complexity \cite{bib:McCormick2005}.

In our series of experiments, we chose $B_i \subset F$ randomly, where $|F|=50$.  We let $n=|E|$ range from $10$ to $190$ in increments of $10$.  For each value of $n$, we ran $10$ experiments.  The average computation time is shown in Figure \ref{fig:exp}, with error bars indicating one standard deviation.  We consistently observed that the computations run in approximately $O(n^{1.85})$ time.  Due to the iterative nature of the SFM algorithm, we anticipate that the computation time could be significantly reduced by implementing the algorithm in C/C++ instead of Matlab. However, the $O(n^{1.85})$ trend should remain the same.  Regardless, we are able to solve the ILP problems under consideration with an astonishing $2^{190}$ constraints in approximately one minute.

\begin{figure}[h!]
 \centering
\includegraphics[width=0.95\textwidth]{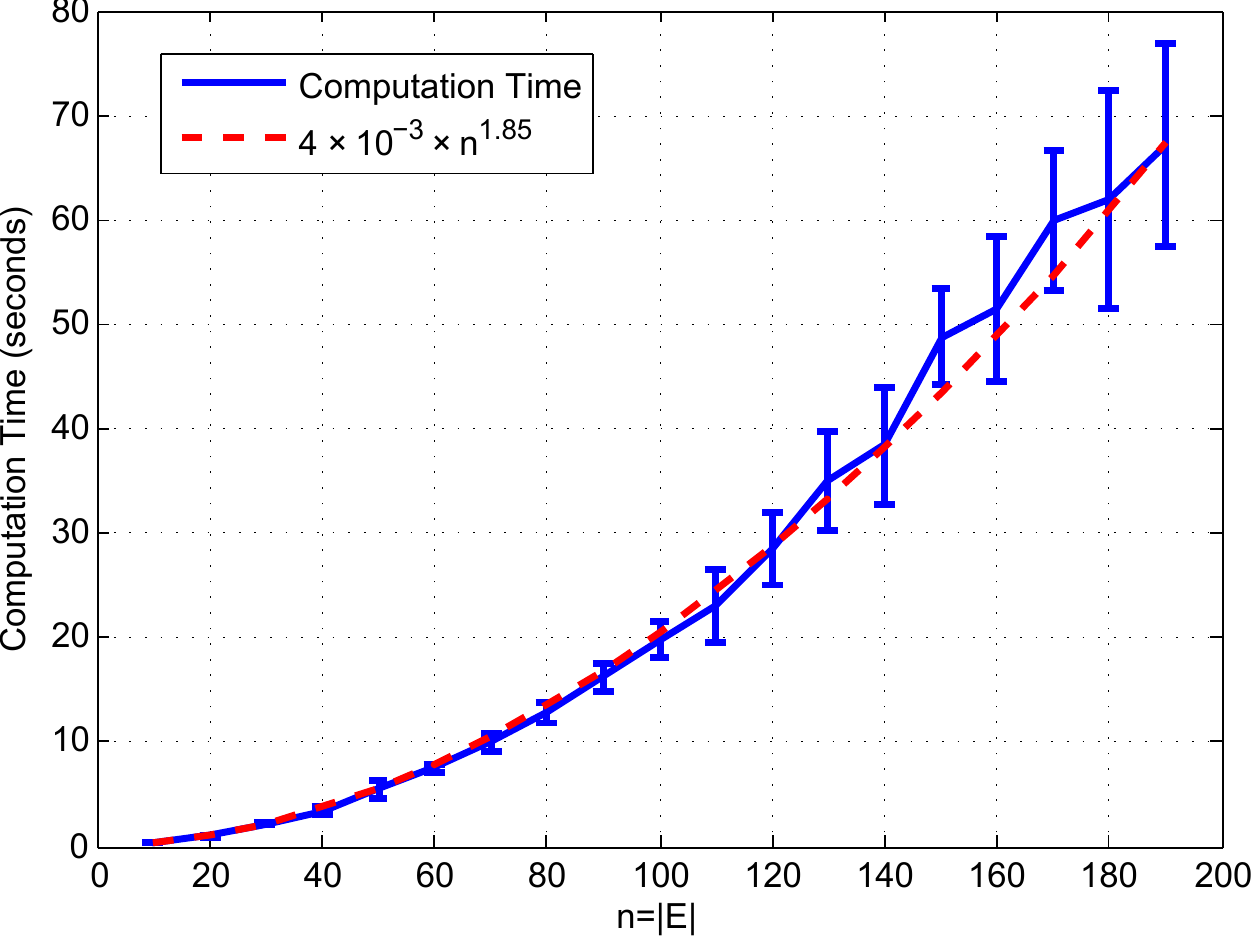}
\caption{Experimental results.  For the red dotted line, the multiplicative constant $\alpha$ and exponent $\beta$ were chosen to minimize the MSE $\sum_{i=1}^n|\log(\alpha n^\beta)-\log(\hat{m}_n)|^2$, where $\hat{m}_n$ is the sample mean of the computation times for $|E|=n$.}\label{fig:exp}
\end{figure}

\section{A Linear Programming Approximation Lemma} \label{app:LPlemma}

\begin{lemma} \label{lem:LP_LSapprox}
Let $A\in\mathbb{R}^{n \times n}$ be a symmetric matrix with nonnegative entries and all column sums equal to $d$. Let $\bar{x}_y$ be the vector of minimum Euclidean norm which minimizes $\|Ax_y-y\|_2$. There exists an optimal solution $x^*$ to the linear program 
\begin{align}
\mbox{minimize~~} & \mathds{1}^T x \label{eqn:lemmaLP}\\
\mbox{subject to:~~} & Ax \succeq y \notag
\end{align}
which satisfies
\begin{align*}
\|x^*-\bar{x}_y\|_{\infty} & \leq c_A \|A\bar{x}_y-y\|_{2},
\end{align*}
where $c_A$ is a constant depending only on $A$.
\end{lemma}

\begin{proof}[Proof of Lemma \ref{lem:LP_LSapprox}]
To begin the proof, we make a few definitions.  Let $\lambda$ be the absolute value of the nonzero eigenvalue of $A$ with smallest modulus (at least one exists since $d$ is an eigenvalue).  Define $\mathcal{N}(A)$ to be the nullspace of $A$, and let $\mathcal{N}^{\perp}(A)$ denote its orthogonal complement.  Finally, let $A^+$ denote the Moore-Penrose pseudoinverse of $A$.

Fix $\bar{x}_y\in\mathbb{R}^n$, and note that $x^*$ is an optimal solution to LP \eqref{eqn:lemmaLP} if and only if $x^*-\bar{x}_y$ is an optimal solution to the linear program
\begin{align*}
\mbox{minimize~~} & \mathds{1}^T (x+\bar{x}_y) \\
\mbox{subject to:~~} & A(x+\bar{x}_y) \succeq y
\end{align*}
with variable $x\in \mathbb{R}^n$.  With this in mind, put $\bar{x}_y=A^+ y$ and define $b=y-A\bar{x}_y$.  By definition of the pseudoinverse, $\bar{x}_y$ is the vector of minimum Euclidean norm which minimizes  $\|Ax_y-y\|_2$.  Moreover, $b\in\mathcal{N}(A)$.

Thus, in order to prove the lemma, it suffices to show the existence of an optimal solution $x^*$ to the linear program 
\begin{align}
\mbox{minimize~~} & \mathds{1}^T x \label{eqn:LPwithoutconstraints} \\
\mbox{subject to:~~} & Ax \succeq b \notag 
\end{align}
which also satisfies the additional constraints
\begin{align}
|x_i| & \leq c_A \|b\|_{2} \mbox{~~for $i=1,\dots, n$,} \label{eqn:extraConstraints}
\end{align}
where $c_A$ is a constant depending only on $A$.

\begin{claim} \label{clm:dualsEqual}
There exists an optimal solution $x^*$ to Linear Program \eqref{eqn:LPwithoutconstraints} which satisfies
\begin{align}
x^*_i  \leq (d \lambda)^{-1} n \|b\|_{\infty} \mbox{~~for $i=1,\dots, n$.} \label{eqn:extraConstraints}
\end{align}
\end{claim}
The proof relies heavily on duality.  The reader is directed to \cite{bib:BoydVandenberghe2004} or any other standard text for details.

To prove the claim, consider LP \eqref{eqn:LPwithoutconstraints}.  By premultiplying the inequality constraint by $d^{-1} \mathds{1}^T$ on both sides, we see that $\mathds{1}^Tx \geq d^{-1} \mathds{1}^Tb > -\infty$.  Hence, the objective is bounded from below, which implies that strong duality holds.  Thus, let $\tilde{z}$ be an optimal solution to the dual LP of \eqref{eqn:LPwithoutconstraints}:
\begin{align}
\mbox{maximize~~} & b^T z \label{eqn:dualLP1} \\
\mbox{subject to:~~} & Az = \mathds{1} \notag \\
& z \succeq 0 \notag
\end{align}
with dual variable $z\in \mathbb{R}^n$.

Next, consider the dual LP of \eqref{eqn:LPwithoutconstraints} with the additional inequality constraints corresponding to \eqref{eqn:extraConstraints}:
\begin{align}
\mbox{maximize~~} & b^T z -(d\lambda)^{-1} n\|b\|_{\infty} \mathds{1}^T y \label{eqn:dualLP2}\\
\mbox{subject to:~~} & Az = \mathds{1}+y \notag\\
& z \succeq 0 \notag\\
& y \succeq 0 \notag
\end{align}
with dual variables $z\in \mathbb{R}^n$ and $y\in \mathbb{R}^n$.  Equivalently, by setting $z=\tilde{z}+\Delta z$ and observing that $y=A\Delta z$, we can write the dual LP \eqref{eqn:dualLP2} as 
\begin{align}
\mbox{maximize~~} & b^T \tilde{z} + b^T \Delta {z} -(d\lambda)^{-1} n\|b\|_{\infty} \mathds{1}^T A \Delta z \label{eqn:dualLP3}\\
\mbox{subject to:~~} 
& A \Delta z \succeq 0 \notag \\
& \tilde{z} +  \Delta z \succeq 0 \notag
\end{align} 
with dual variables $\Delta z\in \mathbb{R}^n$. We prove the claim by showing that the dual LPs \eqref{eqn:dualLP1} and \eqref{eqn:dualLP3} have the same optimal value.  Since strong duality holds, the corresponding primal problems must also have the same optimal value.

Without loss of generality, we can uniquely decompose $\Delta z = \Delta z_1 + \Delta z_2$ where $\Delta z_1 \in \mathcal{N}(A)$ and $\Delta z_2 \in \mathcal{N}^{\perp}(A)$.  Since  $b\in\mathcal{N}(A)$, we have $b^T \Delta z_2 = 0$ and we can rewrite \eqref{eqn:dualLP3} yet again as
\begin{align}
\mbox{maximize~~} & b^T \tilde{z} + b^T \Delta {z}_1 -(d\lambda)^{-1}n \|b\|_{\infty} \mathds{1}^T A \Delta z_2 \label{eqn:dualLP4}\\
\mbox{subject to:~~} 
& A \Delta z_2 \succeq 0 \notag\\
& \tilde{z} +  \Delta z_1 +\Delta z_2 \succeq 0 \label{eqn:ineqConstDz1Dz2}\\
& \Delta z_1 \in \mathcal{N}(A), \Delta z_2 \in \mathcal{N}^{\perp}(A).\notag
\end{align} 

By definition of $\lambda$, for any unit vector $u \in \mathcal{N}^{\perp}(A)$ with $\|u\|_2=1$ we have $\|Au\|_2 \geq \lambda$.
Using this and the fact that $A \Delta z_2 \succeq 0$ for all feasible $\Delta z_2$, we have the following inequality:
\begin{align*}
\mathds{1}^T A \Delta z_2 = \|A \Delta z_2 \|_1 \geq \|A \Delta z_2 \|_2 \geq \lambda \| \Delta z_2 \|_2.
\end{align*}
Thus, the objective \eqref{eqn:dualLP4} can be upper bounded as follows:
\begin{align}
b^T \tilde{z} + b^T \Delta {z}_1 - (d\lambda)^{-1} n \|b\|_{\infty} \mathds{1}^T A \Delta z_2
&\leq b^T \tilde{z} + b^T \Delta {z}_1 - d^{-1} n \|b\|_{\infty} \|\Delta z_2\|_2. \label{eqn:objUB}
\end{align}

Next, we obtain an upper bound on $b^T \Delta {z}_1$.  To this end, 
observe that constraint \eqref{eqn:ineqConstDz1Dz2} implies that $\tilde{z} +  \Delta z_1 \succeq -\mathds{1} \|\Delta z_2 \|_{\infty}$.  Motivated by this, consider the following $\epsilon$-perturbed LP:
\begin{align}
\mbox{minimize~~} &-b^T v \label{eqn:epsPrimal}\\
\mbox{subject to:~~} 
& \tilde{z} +  v \succeq -\epsilon \mathds{1} \notag \\
& v \in \mathcal{N}(A). \notag
\end{align} 
with variable $v$.  Let $p^{*}(\epsilon)$ denote the optimal value of the $\epsilon$-perturbed problem. First observe that $p^{*}(0)=0$.  To see this, note that if $\tilde{z}+v \succeq 0$, then $b^T  v \leq 0$, else we would contradict the optimality of $\tilde{z}$ since $z=\tilde{z}+v$ is a feasible solution to the dual LP \eqref{eqn:dualLP1} in this case.  Now, weak duality implies
\begin{align}
-b^T v \geq p^{*}(\epsilon) \geq p^{*}(0) - \epsilon \mathds{1}^T w^*, \label{eqn:sensAn}
\end{align}
where $w^*$ corresponds to an optimal solution to the dual LP of the unperturbed primal LP \eqref{eqn:epsPrimal}, given by:
\begin{align}
\mbox{maximize~~} &-\tilde{z}^T (Aw-b) \label{eqn:dualLP5}\\
\mbox{subject to:~~} 
& Aw \succeq b. \notag
\end{align} 
Hence, \eqref{eqn:sensAn} implies that
\begin{align}
b^T \Delta {z}_1 \leq \|\Delta z_2 \|_{\infty} \mathds{1}^T w^* \label{eqn:btzUB}
\end{align}
if $\Delta {z}_1, \Delta z_2$ are feasible for LP \eqref{eqn:dualLP4}.

By definition of $\tilde{z}$, $\tilde{z}^TA=\mathds{1}^T$, and hence a vector $w^*$ is optimal for \eqref{eqn:dualLP5} if and only if it also optimizes:
\begin{align*}
\mbox{minimize~~} & \mathds{1}^T w \\
\mbox{subject to:~~} 
& Aw \succeq b.
\end{align*} 
Combining this with \eqref{eqn:btzUB}, we have 
\begin{align*}
b^T \Delta {z}_1 \leq \|\Delta z_2 \|_{\infty} \mathds{1}^T w^* \leq \|\Delta z_2 \|_{\infty} \mathds{1}^T w
\end{align*}
for any vector $w$ satisfying $Aw \succeq b$.  Trivially, $w=d^{-1} \|b\|_{\infty} \mathds{1}$ satisfies this, and hence we obtain:
\begin{align*}
b^T \Delta {z}_1 \leq d^{-1} n \|b\|_{\infty} \|\Delta z_2 \|_{\infty}.
\end{align*}
Finally, we substitute this into \eqref{eqn:objUB} and see that
\begin{align*}
b^T z
&\leq b^T \tilde{z} + d^{-1} n \|b\|_{\infty} \|\Delta z_2 \|_{\infty} - d^{-1}n \|b\|_{\infty} \|\Delta z_2\|_2\\
&\leq b^T \tilde{z} + d^{-1} n \|b\|_{\infty} \|\Delta z_2 \|_{2} - d^{-1}n \|b\|_{\infty} \|\Delta z_2\|_2 \\
&\leq b^T \tilde{z}
\end{align*}
for all vectors $z$ which are feasible for the dual LP \eqref{eqn:dualLP2}.  This completes the proof of Claim \ref{clm:dualsEqual}.

\begin{claim} \label{clm:xLB}
There exists an optimal solution $x^*$ to Linear Program \eqref{eqn:LPwithoutconstraints} which satisfies
\begin{align}
|x_i| \leq c_A \|b\|_{2} \mbox{~~for $i=1,\dots, n$} 
\end{align}
for some constant $c_A$ depending only on $A$.
\end{claim}

First note that $\|b\|_{\infty}\leq \|b\|_2$ for any $b\in\mathbb{R}^n$, hence it suffices to prove the claim for the infinity norm.
Claim \ref{clm:dualsEqual} shows that each of the $x_i$'s can be upper bounded by $(d\lambda)^{-1}n \|b\|_{\infty}$ without affecting the optimal value of LP \eqref{eqn:LPwithoutconstraints}.  To see the lower bound, let $a_j^T$ be a row of $A$ with entry $a_{ji}\geq d/n$ in the $i^{th}$ coordinate (at least one exists for each $i$ since the columns of $A$ sum to $d$). Now, the inequality constraint $Ax \succeq b$ combined with the upper bound on each $x_i$ implies:
\begin{align}
a_{ji} x_i + (d-a_{ji})\lambda^{-1} n \|b\|_{\infty} \geq a_j^T x \geq b_j \geq -\|b\|_{\infty}. \label{eqn:xLB}
\end{align}
Since $a_{ji}\geq d/n$, \eqref{eqn:xLB} implies:
\begin{align*}
x_i  \geq -\lambda^{-1}n(n-1)\|b\|_{\infty}.
\end{align*}
Hence, we can take $c_A= \lambda^{-1}n \times \max\{ n-1, d^{-1}\}$.  This proves Claim \ref{clm:xLB}, and, by our earlier remarks, proves the lemma.
\end{proof}

\bibliographystyle{IEEEtran.bst}
\bibliography{sharingUR}

\end{document}

%% file: figs/G_sets.pdf_tex
\begingroup%
  \makeatletter%
  \providecommand\color[2][]{%
    \errmessage{(Inkscape) Color is used for the text in Inkscape, but the package 'color.sty' is not loaded}%
    \renewcommand\color[2][]{}%
  }%
  \providecommand\transparent[1]{%
    \errmessage{(Inkscape) Transparency is used (non-zero) for the text in Inkscape, but the package 'transparent.sty' is not loaded}%
    \renewcommand\transparent[1]{}%
  }%
  \providecommand\rotatebox[2]{#2}%
  \ifx\svgwidth\undefined%
    \setlength{\unitlength}{313.15683594bp}%
    \ifx\svgscale\undefined%
      \relax%
    \else%
      \setlength{\unitlength}{\unitlength * \real{\svgscale}}%
    \fi%
  \else%
    \setlength{\unitlength}{\svgwidth}%
  \fi%
  \global\let\svgwidth\undefined%
  \global\let\svgscale\undefined%
  \makeatother%
  \begin{picture}(1,0.59329036)%
    \put(0,0){\includegraphics[width=\unitlength]{G_sets.pdf}}%
    \put(0.82843017,0.59121443){\color[rgb]{0,0,0}\makebox(0,0)[lt]{\begin{minipage}{0.40144192\unitlength}\raggedright $\Gamma(S)$\end{minipage}}}%
    \put(0.59907475,0.22137192){\color[rgb]{0,0,0}\makebox(0,0)[lt]{\begin{minipage}{0.40144192\unitlength}\raggedright $S$\end{minipage}}}%
  \end{picture}%
\endgroup%

%% file: figs/lineNet.pdf_tex
\begingroup%
  \makeatletter%
  \providecommand\color[2][]{%
    \errmessage{(Inkscape) Color is used for the text in Inkscape, but the package 'color.sty' is not loaded}%
    \renewcommand\color[2][]{}%
  }%
  \providecommand\transparent[1]{%
    \errmessage{(Inkscape) Transparency is used (non-zero) for the text in Inkscape, but the package 'transparent.sty' is not loaded}%
    \renewcommand\transparent[1]{}%
  }%
  \providecommand\rotatebox[2]{#2}%
  \ifx\svgwidth\undefined%
    \setlength{\unitlength}{495.41401367bp}%
    \ifx\svgscale\undefined%
      \relax%
    \else%
      \setlength{\unitlength}{\unitlength * \real{\svgscale}}%
    \fi%
  \else%
    \setlength{\unitlength}{\svgwidth}%
  \fi%
  \global\let\svgwidth\undefined%
  \global\let\svgscale\undefined%
  \makeatother%
  \begin{picture}(1,0.62015913)%
    \put(0,0){\includegraphics[width=\unitlength]{lineNet.pdf}}%
    \put(0.18407576,0.45588001){\color[rgb]{0,0,0}\makebox(0,0)[lt]{\begin{minipage}{0.17949239\unitlength}\raggedright $p_1$\end{minipage}}}%
    \put(0.85586872,0.44995555){\color[rgb]{0,0,0}\makebox(0,0)[lt]{\begin{minipage}{0.17949239\unitlength}\raggedright $p_2$\end{minipage}}}%
    \put(0.32505299,0.54625532){\color[rgb]{0,0,0}\makebox(0,0)[lt]{\begin{minipage}{0.17949239\unitlength}\raggedright $p_1$\end{minipage}}}%
    \put(0.69851209,0.54567203){\color[rgb]{0,0,0}\makebox(0,0)[lt]{\begin{minipage}{0.17949239\unitlength}\raggedright $p_2$\end{minipage}}}%
    \put(0.18816005,0.15359784){\color[rgb]{0,0,0}\makebox(0,0)[lt]{\begin{minipage}{0.17949239\unitlength}\raggedright $p_1$\end{minipage}}}%
    \put(0.8503347,0.15311895){\color[rgb]{0,0,0}\makebox(0,0)[lt]{\begin{minipage}{0.17949239\unitlength}\raggedright $p_2$\end{minipage}}}%
    \put(0.47305633,0.2548643){\color[rgb]{0,0,0}\makebox(0,0)[lt]{\begin{minipage}{0.2947221\unitlength}\raggedright $p_1\oplus p_2$\end{minipage}}}%
    \put(0.4962351,0.15030532){\color[rgb]{0,0,0}\makebox(0,0)[lt]{\begin{minipage}{0.17949239\unitlength}\raggedright $p_1, p_2$\end{minipage}}}%
    \put(0.1378161,0.37411685){\color[rgb]{0,0,0}\makebox(0,0)[lt]{\begin{minipage}{0.17949239\unitlength}\raggedright Node 1\end{minipage}}}%
    \put(0.4802781,0.37940098){\color[rgb]{0,0,0}\makebox(0,0)[lt]{\begin{minipage}{0.17949239\unitlength}\raggedright Node 2\end{minipage}}}%
    \put(0.8215353,0.37940096){\color[rgb]{0,0,0}\makebox(0,0)[lt]{\begin{minipage}{0.17949239\unitlength}\raggedright Node 3\end{minipage}}}%
    \put(0.13961231,0.06874104){\color[rgb]{0,0,0}\makebox(0,0)[lt]{\begin{minipage}{0.17949239\unitlength}\raggedright Node 1\end{minipage}}}%
    \put(0.4820743,0.07402517){\color[rgb]{0,0,0}\makebox(0,0)[lt]{\begin{minipage}{0.17949239\unitlength}\raggedright Node 2\end{minipage}}}%
    \put(0.82333147,0.07402515){\color[rgb]{0,0,0}\makebox(0,0)[lt]{\begin{minipage}{0.17949239\unitlength}\raggedright Node 3\end{minipage}}}%
    \put(0.02294315,0.60283527){\color[rgb]{0,0,0}\makebox(0,0)[lt]{\begin{minipage}{0.47643045\unitlength}\raggedright Time Instant 1:\end{minipage}}}%
    \put(0.02249501,0.28392772){\color[rgb]{0,0,0}\makebox(0,0)[lt]{\begin{minipage}{0.47643045\unitlength}\raggedright Time Instant 2:\end{minipage}}}%
  \end{picture}%
\endgroup%

%% file: figs/G_SeqSets.pdf_tex
\begingroup%
  \makeatletter%
  \providecommand\color[2][]{%
    \errmessage{(Inkscape) Color is used for the text in Inkscape, but the package 'color.sty' is not loaded}%
    \renewcommand\color[2][]{}%
  }%
  \providecommand\transparent[1]{%
    \errmessage{(Inkscape) Transparency is used (non-zero) for the text in Inkscape, but the package 'transparent.sty' is not loaded}%
    \renewcommand\transparent[1]{}%
  }%
  \providecommand\rotatebox[2]{#2}%
  \ifx\svgwidth\undefined%
    \setlength{\unitlength}{401.92436523bp}%
    \ifx\svgscale\undefined%
      \relax%
    \else%
      \setlength{\unitlength}{\unitlength * \real{\svgscale}}%
    \fi%
  \else%
    \setlength{\unitlength}{\svgwidth}%
  \fi%
  \global\let\svgwidth\undefined%
  \global\let\svgscale\undefined%
  \makeatother%
  \begin{picture}(1,0.49549096)%
    \put(0,0){\includegraphics[width=\unitlength]{G_SeqSets.pdf}}%
    \put(0.87613888,0.08829166){\color[rgb]{0,0,0}\makebox(0,0)[lt]{\begin{minipage}{0.31278094\unitlength}\raggedright $S_0$\end{minipage}}}%
    \put(0.86807829,0.39422593){\color[rgb]{0,0,0}\makebox(0,0)[lt]{\begin{minipage}{0.31278094\unitlength}\raggedright $S_2$\end{minipage}}}%
    \put(0.87354109,0.2390012){\color[rgb]{0,0,0}\makebox(0,0)[lt]{\begin{minipage}{0.31278094\unitlength}\raggedright $S_1$\end{minipage}}}%
  \end{picture}%
\endgroup%

%% file: figs/G_NC.pdf_tex
\begingroup%
  \makeatletter%
  \providecommand\color[2][]{%
    \errmessage{(Inkscape) Color is used for the text in Inkscape, but the package 'color.sty' is not loaded}%
    \renewcommand\color[2][]{}%
  }%
  \providecommand\transparent[1]{%
    \errmessage{(Inkscape) Transparency is used (non-zero) for the text in Inkscape, but the package 'transparent.sty' is not loaded}%
    \renewcommand\transparent[1]{}%
  }%
  \providecommand\rotatebox[2]{#2}%
  \ifx\svgwidth\undefined%
    \setlength{\unitlength}{234.17143555bp}%
    \ifx\svgscale\undefined%
      \relax%
    \else%
      \setlength{\unitlength}{\unitlength * \real{\svgscale}}%
    \fi%
  \else%
    \setlength{\unitlength}{\svgwidth}%
  \fi%
  \global\let\svgwidth\undefined%
  \global\let\svgscale\undefined%
  \makeatother%
  \begin{picture}(1,2.24438747)%
    \put(0,0){\includegraphics[width=\unitlength]{G_NC.pdf}}%
    \put(0.55636895,2.17179111){\color[rgb]{0,0,0}\makebox(0,0)[lb]{\smash{$s$}}}%
    \put(0.19247909,1.80922611){\color[rgb]{0,0,0}\makebox(0,0)[lb]{\smash{$u_1$}}}%
    \put(0.9277157,1.81421419){\color[rgb]{0,0,0}\makebox(0,0)[lb]{\smash{$u_2$}}}%
    \put(0.19760352,1.46770386){\color[rgb]{0,0,0}\makebox(0,0)[lb]{\smash{$v_1^0$}}}%
    \put(0.18383354,0.74832842){\color[rgb]{0,0,0}\makebox(0,0)[lb]{\smash{$v_1^1$}}}%
    \put(0.19359439,0.04066629){\color[rgb]{0,0,0}\makebox(0,0)[lb]{\smash{$v_1^2$}}}%
    \put(0.56358324,1.46526361){\color[rgb]{0,0,0}\makebox(0,0)[lb]{\smash{$v_2^0$}}}%
    \put(0.55469364,0.74588818){\color[rgb]{0,0,0}\makebox(0,0)[lb]{\smash{$v_2^1$}}}%
    \put(0.55957408,0.03822625){\color[rgb]{0,0,0}\makebox(0,0)[lb]{\smash{$v_2^2$}}}%
    \put(0.93059163,1.45843101){\color[rgb]{0,0,0}\makebox(0,0)[lb]{\smash{$v_3^0$}}}%
    \put(0.92560637,0.75076867){\color[rgb]{0,0,0}\makebox(0,0)[lb]{\smash{$v_3^1$}}}%
    \put(0.93048675,0.03822625){\color[rgb]{0,0,0}\makebox(0,0)[lb]{\smash{$v_3^2$}}}%
    \put(0.19115417,1.11558073){\color[rgb]{0,0,0}\makebox(0,0)[lb]{\smash{$w_1^1$}}}%
    \put(0.19115417,0.39522945){\color[rgb]{0,0,0}\makebox(0,0)[lb]{\smash{$w_1^2$}}}%
    \put(0.55713389,1.11021211){\color[rgb]{0,0,0}\makebox(0,0)[lb]{\smash{$w_2^1$}}}%
    \put(0.55713389,0.39571737){\color[rgb]{0,0,0}\makebox(0,0)[lb]{\smash{$w_2^2$}}}%
    \put(0.92902267,1.1150926){\color[rgb]{0,0,0}\makebox(0,0)[lb]{\smash{$w_3^1$}}}%
    \put(0.91828574,0.39571737){\color[rgb]{0,0,0}\makebox(0,0)[lb]{\smash{$w_3^2$}}}%
    \put(0.13990966,1.29981721){\color[rgb]{0,0,0}\makebox(0,0)[lb]{\smash{$b_1^1$}}}%
    \put(0.50447388,1.30469769){\color[rgb]{0,0,0}\makebox(0,0)[lb]{\smash{$b_2^1$}}}%
    \put(0.8753865,1.30957777){\color[rgb]{0,0,0}\makebox(0,0)[lb]{\smash{$b_3^1$}}}%
    \put(0.13356123,0.59703514){\color[rgb]{0,0,0}\makebox(0,0)[lb]{\smash{$b_1^2$}}}%
    \put(0.50935426,0.59996373){\color[rgb]{0,0,0}\makebox(0,0)[lb]{\smash{$b_2^2$}}}%
    \put(0.88026699,0.60679591){\color[rgb]{0,0,0}\makebox(0,0)[lb]{\smash{$b_3^2$}}}%
    \put(0.28095021,2.04652263){\color[rgb]{0,0,0}\makebox(0,0)[lb]{\smash{$1$}}}%
    \put(0.67658529,2.04836714){\color[rgb]{0,0,0}\makebox(0,0)[lb]{\smash{$1$}}}%
  \end{picture}%
\endgroup%

%% file: figs/G_NC_cut.pdf_tex
\begingroup%
  \makeatletter%
  \providecommand\color[2][]{%
    \errmessage{(Inkscape) Color is used for the text in Inkscape, but the package 'color.sty' is not loaded}%
    \renewcommand\color[2][]{}%
  }%
  \providecommand\transparent[1]{%
    \errmessage{(Inkscape) Transparency is used (non-zero) for the text in Inkscape, but the package 'transparent.sty' is not loaded}%
    \renewcommand\transparent[1]{}%
  }%
  \providecommand\rotatebox[2]{#2}%
  \ifx\svgwidth\undefined%
    \setlength{\unitlength}{329.2222168bp}%
    \ifx\svgscale\undefined%
      \relax%
    \else%
      \setlength{\unitlength}{\unitlength * \real{\svgscale}}%
    \fi%
  \else%
    \setlength{\unitlength}{\svgwidth}%
  \fi%
  \global\let\svgwidth\undefined%
  \global\let\svgscale\undefined%
  \makeatother%
  \begin{picture}(1,1.62408752)%
    \put(0,0){\includegraphics[width=\unitlength]{G_NC_cut.pdf}}%
    \put(0.50313525,1.57245069){\color[rgb]{0,0,0}\makebox(0,0)[lb]{\smash{$s$}}}%
    \put(0.24430518,1.31456296){\color[rgb]{0,0,0}\makebox(0,0)[lb]{\smash{$u_1$}}}%
    \put(0.76726931,1.31811092){\color[rgb]{0,0,0}\makebox(0,0)[lb]{\smash{$u_2$}}}%
    \put(0.24795012,1.07164267){\color[rgb]{0,0,0}\makebox(0,0)[lb]{\smash{$v_1^0$}}}%
    \put(0.23815571,0.55996038){\color[rgb]{0,0,0}\makebox(0,0)[lb]{\smash{$v_1^1$}}}%
    \put(0.24509847,0.05660961){\color[rgb]{0,0,0}\makebox(0,0)[lb]{\smash{$v_1^2$}}}%
    \put(0.50826668,1.06990696){\color[rgb]{0,0,0}\makebox(0,0)[lb]{\smash{$v_2^0$}}}%
    \put(0.50194363,0.55822467){\color[rgb]{0,0,0}\makebox(0,0)[lb]{\smash{$v_2^1$}}}%
    \put(0.50541502,0.05487404){\color[rgb]{0,0,0}\makebox(0,0)[lb]{\smash{$v_2^2$}}}%
    \put(0.76931493,1.06504702){\color[rgb]{0,0,0}\makebox(0,0)[lb]{\smash{$v_3^0$}}}%
    \put(0.76576897,0.5616961){\color[rgb]{0,0,0}\makebox(0,0)[lb]{\smash{$v_3^1$}}}%
    \put(0.76924033,0.05487404){\color[rgb]{0,0,0}\makebox(0,0)[lb]{\smash{$v_3^2$}}}%
    \put(0.24336278,0.82118212){\color[rgb]{0,0,0}\makebox(0,0)[lb]{\smash{$w_1^1$}}}%
    \put(0.24336278,0.30880573){\color[rgb]{0,0,0}\makebox(0,0)[lb]{\smash{$w_1^2$}}}%
    \put(0.50367934,0.81736349){\color[rgb]{0,0,0}\makebox(0,0)[lb]{\smash{$w_2^1$}}}%
    \put(0.50367934,0.30915278){\color[rgb]{0,0,0}\makebox(0,0)[lb]{\smash{$w_2^2$}}}%
    \put(0.76819894,0.82083492){\color[rgb]{0,0,0}\makebox(0,0)[lb]{\smash{$w_3^1$}}}%
    \put(0.76056191,0.30915278){\color[rgb]{0,0,0}\makebox(0,0)[lb]{\smash{$w_3^2$}}}%
    \put(0.20691323,0.95222711){\color[rgb]{0,0,0}\makebox(0,0)[lb]{\smash{$b_1^1$}}}%
    \put(0.46622297,0.95569853){\color[rgb]{0,0,0}\makebox(0,0)[lb]{\smash{$b_2^1$}}}%
    \put(0.73004824,0.95916966){\color[rgb]{0,0,0}\makebox(0,0)[lb]{\smash{$b_3^1$}}}%
    \put(0.20239768,0.45234746){\color[rgb]{0,0,0}\makebox(0,0)[lb]{\smash{$b_1^2$}}}%
    \put(0.46969432,0.45443053){\color[rgb]{0,0,0}\makebox(0,0)[lb]{\smash{$b_2^2$}}}%
    \put(0.73351967,0.45929017){\color[rgb]{0,0,0}\makebox(0,0)[lb]{\smash{$b_3^2$}}}%
    \put(0.30723352,1.48334886){\color[rgb]{0,0,0}\makebox(0,0)[lb]{\smash{$1$}}}%
    \put(0.58864355,1.48466084){\color[rgb]{0,0,0}\makebox(0,0)[lb]{\smash{$1$}}}%
    \put(0.04699523,0.39287394){\color[rgb]{0,0,0}\makebox(0,0)[lt]{\begin{minipage}{0.26035388\unitlength}\raggedright $S$\end{minipage}}}%
    \put(0.04784917,0.48333366){\color[rgb]{0,0,0}\makebox(0,0)[lt]{\begin{minipage}{0.26035388\unitlength}\raggedright $S^c$\end{minipage}}}%
    \put(0.07407999,1.42606004){\color[rgb]{0,0,0}\makebox(0,0)[lt]{\begin{minipage}{0.26035388\unitlength}\raggedright $S'$\end{minipage}}}%
    \put(0.07493393,1.51651977){\color[rgb]{0,0,0}\makebox(0,0)[lt]{\begin{minipage}{0.26035388\unitlength}\raggedright $S'^c$\end{minipage}}}%
  \end{picture}%
\endgroup%

%% file: sharingJournalclean.bbl
\begin{thebibliography}{10}
\providecommand{\url}[1]{#1}
\csname url@samestyle\endcsname
\providecommand{\newblock}{\relax}
\providecommand{\bibinfo}[2]{#2}
\providecommand{\BIBentrySTDinterwordspacing}{\spaceskip=0pt\relax}
\providecommand{\BIBentryALTinterwordstretchfactor}{4}
\providecommand{\BIBentryALTinterwordspacing}{\spaceskip=\fontdimen2\font plus
\BIBentryALTinterwordstretchfactor\fontdimen3\font minus
  \fontdimen4\font\relax}
\providecommand{\BIBforeignlanguage}[2]{{%
\expandafter\ifx\csname l@#1\endcsname\relax
\typeout{** WARNING: IEEEtran.bst: No hyphenation pattern has been}%
\typeout{** loaded for the language `#1'. Using the pattern for}%
\typeout{** the default language instead.}%
\else
\language=\csname l@#1\endcsname
\fi
#2}}
\providecommand{\BIBdecl}{\relax}
\BIBdecl

\bibitem{bib:CourtadeWeselMILCOM2010}
T.~Courtade, B.~Xie, and R.~Wesel, ``Optimal exchange of packets for universal
  recovery in broadcast networks,'' in \emph{MILITARY COMMUNICATIONS
  CONFERENCE, 2010 - MILCOM 2010}, 31 2010-nov. 3 2010, pp. 2250 --2255.

\bibitem{bib:CourtadeWeselAllerton2010}
T.~Courtade and R.~Wesel, ``Efficient universal recovery in broadcast
  networks,'' in \emph{Communication, Control, and Computing (Allerton), 2010
  48th Annual Allerton Conference on}, 29 2010-oct. 1 2010, pp. 1542 --1549.

\bibitem{bib:CourtadeWeselAllerton2011}
------, ``Weighted universal recovery, practical secrecy, and an efficient
  algorithm for solving both,'' in \emph{Communication, Control, and Computing
  (Allerton), 2011 49th Annual Allerton Conference on}, Oct. 2011.

\bibitem{bib:Ahlswede00networkinformation}
R.~Ahlswede, N.~Cai, S.~yen Robert~Li, and R.~W. Yeung, ``Network information
  flow,'' \emph{IEEE TRANSACTIONS ON INFORMATION THEORY}, vol.~46, no.~4, pp.
  1204--1216, 2000.

\bibitem{bib:Li03linearnetwork}
S.~yen Robert~Li, R.~W. Yeung, and N.~Cai, ``Linear network coding,''
  \emph{IEEE Transactions on Information Theory}, vol.~49, pp. 371--381, 2003.

\bibitem{bib:ElRouayhebChaudhryITW2007}
S.~El~Rouayheb, M.~Chaudhry, and A.~Sprintson, ``On the minimum number of
  transmissions in single-hop wireless coding networks,'' in \emph{Information
  Theory Workshop, 2007. ITW '07. IEEE}, sept. 2007, pp. 120 --125.

\bibitem{bib:ElRouayheb2010ITW}
S.~El~Rouayheb, A.~Sprintson, and P.~Sadeghi, ``On coding for cooperative data
  exchange,'' in \emph{Information Theory Workshop (ITW), 2010 IEEE}, Jan.
  2010, pp. 1 --5.

\bibitem{bib:SprintsonSadeghiISIT2010}
A.~Sprintson, P.~Sadeghi, G.~Booker, and S.~El~Rouayheb, ``A randomized
  algorithm and performance bounds for coded cooperative data exchange,'' in
  \emph{Information Theory Proceedings (ISIT), 2010 IEEE International
  Symposium on}, june 2010, pp. 1888 --1892.

\bibitem{bib:SprintsonQshine2010}
------, ``Deterministic algorithm for coded cooperative data exchange,'' in
  \emph{ICST QShine}, Nov. 2010.

\bibitem{bib:OzgulSprintsonITA2011}
D.~Ozgul and A.~Sprintson, ``An algorithm for cooperative data exchange with
  cost criterion,'' in \emph{Information Theory and Applications Workshop
  (ITA), 2011}, feb. 2011, pp. 1 --4.

\bibitem{bib:BirkKolIT2006}
Y.~Birk and T.~Kol, ``Coding on demand by an informed source (iscod) for
  efficient broadcast of different supplemental data to caching clients,''
  \emph{Information Theory, IEEE Transactions on}, vol.~52, no.~6, pp. 2825 --
  2830, june 2006.

\bibitem{bib:LubetzkyStavIT2009}
E.~Lubetzky and U.~Stav, ``Nonlinear index coding outperforming the linear
  optimum,'' \emph{Information Theory, IEEE Transactions on}, vol.~55, no.~8,
  pp. 3544 --3551, aug. 2009.

\bibitem{bib:AlonLubetzkyFOCS2008}
N.~Alon, E.~Lubetzky, U.~Stav, A.~Weinstein, and A.~Hassidim, ``Broadcasting
  with side information,'' in \emph{Foundations of Computer Science, 2008. FOCS
  '08. IEEE 49th Annual IEEE Symposium on}, oct. 2008, pp. 823 --832.

\bibitem{bib:CsiszarNarayanIT2004}
I.~Csiszar and P.~Narayan, ``Secrecy capacities for multiple terminals,''
  \emph{Information Theory, IEEE Transactions on}, vol.~50, no.~12, pp. 3047 --
  3061, dec. 2004.

\bibitem{bib:YeNarayanISIT2005}
C.~Ye and P.~Narayan, ``Secret key and private key constructions for simple
  multiterminal source models,'' in \emph{Information Theory, 2005. ISIT 2005.
  Proceedings. International Symposium on}, sept. 2005, pp. 2133 --2137.

\bibitem{bib:YeReznikCISS2010}
C.~Ye and A.~Reznik, ``A simple secret key construction system for broadcasting
  model,'' in \emph{Information Sciences and Systems (CISS), 2010 44th Annual
  Conference on}, march 2010, pp. 1 --6.

\bibitem{bib:Karp1972}
R.~M. Karp, ``{Reducibility Among Combinatorial Problems},'' in
  \emph{Complexity of Computer Computations}, R.~E. Miller and J.~W. Thatcher,
  Eds.\hskip 1em plus 0.5em minus 0.4em\relax Plenum Press, 1972, pp. 85--103.

\bibitem{bib:JaggiIT2005}
S.~Jaggi, P.~Sanders, P.~Chou, M.~Effros, S.~Egner, K.~Jain, and L.~Tolhuizen,
  ``Polynomial time algorithms for multicast network code construction,''
  \emph{Information Theory, IEEE Transactions on}, vol.~51, no.~6, pp. 1973 --
  1982, june 2005.

\bibitem{bib:randLinNC}
T.~Ho, M.~Medard, R.~Koetter, D.~Karger, M.~Effros, J.~Shi, and B.~Leong, ``A
  random linear network coding approach to multicast,'' \emph{Information
  Theory, IEEE Transactions on}, vol.~52, no.~10, pp. 4413 --4430, oct. 2006.

\bibitem{bib:HalfordChuggPIMRC2010}
T.~Halford, K.~Chugg, and A.~Polydoros, ``Barrage relay networks: System amp;
  protocol design,'' in \emph{Personal Indoor and Mobile Radio Communications
  (PIMRC), 2010 IEEE 21st International Symposium on}, sept. 2010, pp. 1133
  --1138.

\bibitem{bib:HalfordChuggITA2010}
T.~Halford and K.~Chugg, ``Barrage relay networks,'' in \emph{Information
  Theory and Applications Workshop (ITA), 2010}, 31 2010-feb. 5 2010, pp. 1
  --8.

\bibitem{bib:HalfordHwangMilcom2010}
T.~Halford and G.~Hwang, ``Barrage relay networks for unmanned ground
  systems,'' in \emph{MILITARY COMMUNICATIONS CONFERENCE, 2010 - MILCOM 2010},
  31 2010-nov. 3 2010, pp. 1274 --1280.

\bibitem{bib:BlairBrownChuggHalfordMILCOM2008}
A.~Blair, T.~Brown, K.~Chugg, T.~Halford, and M.~Johnson, ``Barrage relay
  networks for cooperative transport in tactical manets,'' in \emph{Military
  Communications Conference, 2008. MILCOM 2008. IEEE}, nov. 2008, pp. 1 --7.

\bibitem{bib:RamamoorthyIT2005}
A.~Ramamoorthy, J.~Shi, and R.~D. Wesel, ``On the capacity of network coding
  for random networks,'' \emph{IEEE TRANSACTIONS ON INFORMATION THEORY},
  vol.~51, no.~8, pp. 2878--2885, 2005.

\bibitem{bib:McCormick2005}
S.~McCormick, \emph{Submodular Function Minimization. In Discrete Optimization,
  K. Aardal, G. Nemhauser, and R. Weismantel, eds. Handbooks in Operations
  Research and Management Science}.\hskip 1em plus 0.5em minus 0.4em\relax
  Elsevier, 2005, vol.~12.

\bibitem{bib:Fujishige2010}
S.~Fujishige, \emph{Submodular Functions and Optimization}, 2nd~ed.\hskip 1em
  plus 0.5em minus 0.4em\relax Berlin: Elsevier Science, 2010.

\bibitem{bib:Schrijver2003}
A.~Schrijver, \emph{Combinatorial Optimization: Polyhedra and
  Efficiency}.\hskip 1em plus 0.5em minus 0.4em\relax Berlin: Springer-Verlag,
  2003.

\bibitem{bib:Fujishige06theMinNormPoint}
S.~Fujishige, T.~Hayashi, and S.~Isotani, ``The minimum-norm-point algorithm
  applied to submodular function minimization and,'' in \emph{Kyoto University,
  Kyoto Japan}, 2006.

\bibitem{bib:Krause_sfo}
A.~Krause and S.~Sonnenburg, ``Sfo: A toolbox for submodular function
  optimization, the,'' \emph{Journal of Machine Learning Research}, pp.
  1141--1144, 2010.

\bibitem{bib:BoydVandenberghe2004}
S.~Boyd and L.~Vandenberghe, \emph{Convex Optimization}.\hskip 1em plus 0.5em
  minus 0.4em\relax Cambridge University Press, 2004.

\end{thebibliography}
